\documentclass[12pt, titlepage]{article}
\usepackage{amsmath, amsthm, amssymb}

\usepackage{subfigure}
\usepackage{graphicx}

\usepackage{dsfont}
\usepackage{mathrsfs}
\usepackage{ifsym, wasysym}
\newcommand{\argmax}{\operatornamewithlimits{argmax}}
\newcommand{\argmin}{\operatornamewithlimits{argmin}}
\def\eqd{\,{\buildrel \mathscr{D} \over =}\,}
\def\endremark{\hfill$\square$}

\newtheorem{thm}{Theorem}
\newtheorem{cor}{Corollary}
\newtheorem{lem}{Lemma}
\newtheorem{alg}{Algorithm}
\newtheorem{lemA}{Lemma}[subsection]
\newtheorem{rem}{Remark}

\begin{document}

\title{Optimal   Multi-Server  Allocation to Parallel Queues With Independent Random Queue-Server Connectivity}

\author{Hussein~Al-Zubaidy\thanks{H. Al-Zubaidy is with the Edward S. Rogers Sr. Department of Electrical and Computer Engineering, University of Toronto, Toronto, ON,  M5S 1A1 Canada,  e-mail: hzubaidy@comm.utoronto.ca.},
        Ioannis~Lambadaris\thanks{I. Lambadaris is with the Department of Systems and Computer Engineering, Carleton University, Ottawa, ON, K1S 5B6 Canada,  e-mail:   ioannis.lambadaris@sce.carleton.ca.},
       Yannis~Viniotis\thanks{Y. Viniotis is with the Electrical and Computer Engineering Department, North Carolina State University, Raleigh, NC, USA, e-mail:  candice@ncsu.edu. \newline
       This work was supported by the Natural Sciences and Engineering Research Council of Canada (NSERC).}
}

%

%

\maketitle




\begin{abstract}
We investigate an optimal scheduling problem in a discrete-time system of $L$ parallel queues that are served by $K$ identical, randomly connected servers.
Each queue may be connected to a subset of the $K$ servers during any given time slot.
This model has been widely used in studies of emerging 3G/4G wireless systems. We introduce the class of Most Balancing (MB) policies and provide their mathematical characterization. We prove that MB policies are optimal; we define optimality as minimization, in stochastic ordering sense, of a range of cost functions of the queue lengths, including the process of total number of packets in the system. We use stochastic coupling arguments for our proof.  We  introduce the Least Connected Server First/Longest Connected Queue (LCSF/LCQ) policy as an  easy-to-implement approximation of MB policies. 
We conduct a simulation study to compare the performance of several policies. The simulation results show that: (a) in all cases, LCSF/LCQ approximations to the MB policies outperform the other policies, (b) randomized policies perform fairly close to the optimal one, and, (c) the performance advantage of the optimal policy over the other simulated policies increases as the channel connectivity probability decreases and as the number of servers in the system increases.
\end{abstract}


\section{Introduction, Model Description and Prior Research}\label{YV:intro_section}
Emerging 3G/4G wireless networks can be categorized as high-speed, IP-based, packet access networks. They utilize the channel variability, using data rate adaptation, and user diversity to increase their channel capacity. These systems usually employ a mixture of Time and Code Division Multiple Access (TDMA/CDMA) schemes. Time is divided into equal size slots, each of which can be allocated to one or more users. To optimize the use of the enhanced data rate, these systems allow several users to share the wireless channel simultaneously using CDMA. This will minimize the wasted capacity resulting from the allocation of the whole channel capacity to one user at a time even when that user is unable to utilize all of that capacity. Another reason for sharing system capacity between several users, at the same time slot, is that some of the user equipments at the receiving side might have design limitations on the amount of data they can receive and process at a given time.

The connectivity of users to the base station in any wireless system is varying with time and can be best modeled as a random process. The application of stochastic modeling and queuing theory to model wireless systems is well vetted in the literature. Modeling wireless systems using parallel queues with random queue/server connectivity was used  by Tassiulas and Ephremides \cite{Tassiulas}, Ganti, Modiano and Tsitsiklis \cite{ganti} and many others to study scheduler optimization in wireless systems. In the following subsection, we provide a more formal model description and motivation for the problem at hand.
\subsection{Model Description}\label{Formulation}
In this work, we assume that time is slotted into equal length deterministic intervals. We model the wireless system under investigation as a set of $L$ parallel queues with infinite capacity (see Figure \ref{fig_1}); the queues correspond to the different users in the system. We define $X_i(n)$ to represent the number of packets in the $i^{th}$ queue at the beginning of time slot $n$. The queues share a set of $K$ identical servers, each server representing a network resource, e.g.,  transmission channel. We make no assumption regarding the number of servers relative to the number of queues, i.e., $K$ can be less, equal or greater than $L$. The packets in this system are assumed to have constant length, and require one time slot to complete service. A server can serve one packet only at any given time slot. A server can only serve connected, non-empty queues. Therefore, the system can serve up to $K$ packets during each time slot. Those packets may belong to one or several queues.

The channel connectivity between a queue  and any server is random. The state of the channel connecting the $i^{th}$ queue to the $j^{th}$ server during the $n^{th}$ time slot is denoted by $G_{i,j}(n)$ and can be either connected ($G_{i,j}(n)=1$) or not connected ($G_{i,j}(n)=0$). Therefore, in a real system $G_{i,j}(n)$ will determine if  transmission channel $j$ can be used by user $i$ or not. We assume that $G_{i,j}(n)$, for all $i=1,2,\ldots,L$, $j=1,2,\ldots,K$ and $n$, are independent, Bernoulli random variables with parameter $p$.

The number of arrivals to the $i^{th}$ queue during time slot $n$ is denoted by $Z_i(n)$. 
The random variables $Z_i(n), \forall i,n$ is assumed to have Bernoulli distribution. 
We  require that arrival processes to different queues be independent of each other; we further require that the random processes $\{Z_i(n)\}$ be independent of the processes $\{G_{i,j}(n)\}$ for $i=1,2,\ldots,L$, $j=1,2,\ldots,K$.
The symmetry and independence assumptions are necessary for the coupling arguments we use in our optimality proofs. The rest are simplifying assumptions that can be relaxed at the price of a more complex and maybe less intuitive proof.

A scheduling policy  (or server allocation policy or scheduler) decides, at the beginning of each time slot, what servers will be assigned to which queue during that time slot. The objective of this work is to identify and analyze the optimal scheduling policy that minimizes, in a stochastic ordering sense, a range of cost functions of the system queue sizes, including the total number of queued packets, in the aforementioned system. The choice of the class of cost functions and the minimization process are discussed in detail in Section \ref{YV:main-result}.


\begin{figure}
\centering
\includegraphics [width=3.0in]{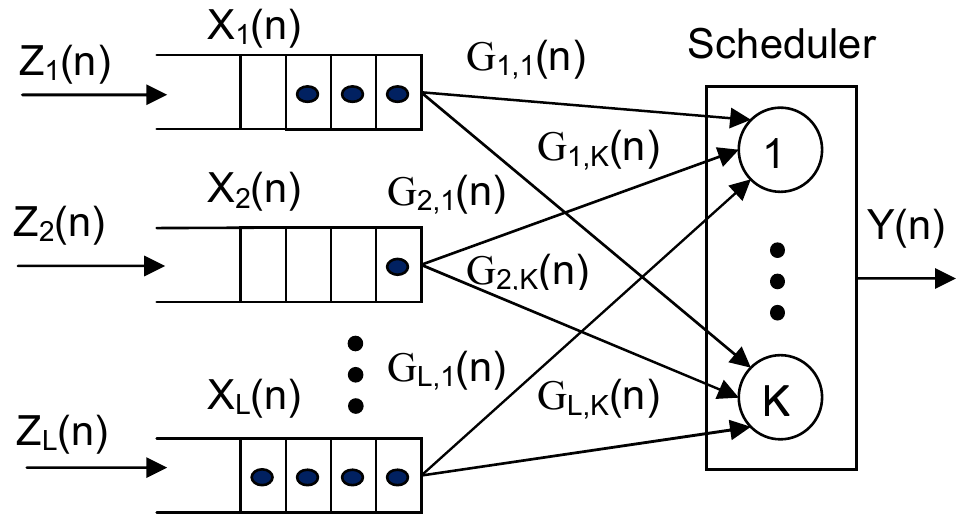}
\caption{Abstraction of downlink scheduler in a multi-server wireless network.}
\label{fig_1}
\vspace{-2mm}
\end{figure}


\subsection{Previous Work and Our Contributions}
In the literature, there is substantial research effort focusing on the subject of optimal scheduling in wireless networks with random connectivity.
Tassiulas and Ephremides \cite{Tassiulas} studied the problem of allocating a single, randomly connected server to a set of parallel queues.
They proved, using stochastic coupling arguments, that a LCQ (Longest Connected Queue) policy is optimal. 
In our work we investigate a more general model that studies the optimal allocation of $K>1$ randomly connected servers to parallel queues. We show that LCQ is not always optimal in a multi-server system where  multiple servers can be allocated to each queue at any given time slot. 
Bambos and Michailidis \cite{Bambos} worked on a similar model (a continuous time version of \cite{Tassiulas} with finite buffer capacity) and proved that under stationary ergodic input job flow and modulation processes, both `Maximum Connected Workload' and LCQ dynamic allocation policies maximize the stability region for this system. Furthermore, they proved that  a policy that allocates the server to the connected queue with the fewest empty spaces, stochastically minimizes the loss flow and maximizes the throughput \cite{Bambos2}.

Another relevant result is that reported by Ganti, Modiano and Tsitsiklis \cite{ganti}. They presented a model for a satellite node that has $K$ transmitters. The system was modeled by a set of parallel queues with symmetrical statistics competing for $K$ identical, randomly connected servers. At each time slot, no more than one server is allocated to each scheduled queue. They proved, using stochastic coupling arguments, that  a policy that allocates the $K$ servers to the $K$ longest connected queues at each time slot, is optimal. This model is similar to the one we consider in this work, except that in our model one or more servers can be allocated to each queue in the system. A further, stronger difference between the two models is that we consider the case where each queue has \emph{independent connectivities} to different servers. We make these assumptions for a more suitable representation of the 3G/4G wireless systems described earlier. These differences make it substantially harder to identify (and even describe) the optimal policy (see Section \ref{sec:MBPolicies}). A more recent  result that has relevance to our work is the one reported by Kittipiyakul and Javidi in \cite{Kittipiyakul}. They proved, using dynamic programming, that a `maximum-throughput and load-balancing'  policy minimizes the expected average cost for a two-queue, multi-server system with random connectivity. In our research work, we prove optimality of the most balancing policies in the more general problem of a multi-queue (more than two queues) and multi-server system with random channel connectivity. A stronger distinction of our work is that we proved the optimality in a stochastic ordering sense which is a stronger notion of optimality compared to the expected average cost criterion that was used in \cite{Kittipiyakul}.
Lott and Teneketzis \cite{lott} investigated a multi-class system of $N$ weighted cost parallel queues and $M$ servers with random connectivity. They also used the same restriction of one server per queue used in \cite{ganti}. They showed that an index rule is optimal and provided conditions sufficient, but not necessary, to guarantee its optimality.

Koole et al \cite{Koole} studied a model similar to that of \cite{Tassiulas} and \cite{Bambos2}. They found that the `Best User'  policy maximizes the expected discounted number of successful transmissions. Liu et al \cite{xin_liu}, \cite{xin_liu1} studied the optimality of opportunistic schedulers (e.g., Proportional Fair (PF) scheduler). They presented the characteristics and optimality conditions for such schedulers. However, Andrews \cite{Andrews} showed that there are six different implementation algorithms of a PF scheduler, none of which is stable. For more information on resource allocation and optimization in wireless networks the reader may consult \cite{Agrawal}, \cite{Andrews1}, \cite{Eryilmaz}, \cite{Stolyar}, \cite{Lu}, and \cite{Shakkottai}.

The model we present in this work can be applied to many of the previous work described above. In section \ref{sec:final-remarks} we discuss this applicability for three key publications, namely \cite{Tassiulas}, \cite{ganti} and  \cite{Kittipiyakul}, that are strongly related to our own. We also show how our model can be reduced to their models and used to describe the problems they investigated.

In summary, the main contributions of our work are the following:
\begin{enumerate}
\item We introduce and show the existence of the class of Most Balancing (MB) scheduling policies  in the model of Figure \ref{fig_1} (see Equations (\ref{eq:kappa}) and (\ref{eq:MBpolicy})). Intuitively, an MB policy attempts to balance all queue sizes at every time slot, so that the total sum of queue size differences will be minimized. 

\item We prove the optimality of MB policies for minimizing, in stochastic ordering sense, a set of functionals of the queue lengths (see Theorem \ref{thm:3}).

\item We provide low-overhead, heuristic approximations for an MB policy. At any time slot,  such policies allocate the ``least connected servers first'' to their ``longest connected queues'' (LCSF/LCQ). These  policies have  $O(L \times K)$ complexity and thus can be easily implemented.
 We evaluate the performance of these approximations via simulations. 

\end{enumerate}

The rest of the article is organized as follows. In section II, we introduce notation and define the scheduling policies. In section III, we introduce and provide a detailed description of the MB  policies.
In section IV, we introduce and characterize \textit{balancing interchanges}, that we will use in  the proof of MB optimality. In section V, we present the main result, i.e., the optimality of MB  policies. In section VI, we present the Least Balancing (LB) policies, and show that these policies perform the worst among all work conserving policies. MB and LB policies provide upper and lower performance bounds. In section VII, we  introduce practical, low-overhead approximations for such policies, namely the LCSF/LCQ policy and the MCSF/SCQ policy, with their implementation algorithms. In section VIII, we present simulation results for different scheduling policies. In section IX, we give some final remarks that show the applicability of our model to problems studied in previous work.  We present proofs for some of our results in the Appendix.

\section{Scheduling Policies} \label{Policies-section}
Recall that $L$ and $K$ denote the number of queues and servers respectively in the model introduced in Figure \ref{fig_1}. We will use \textbf{bold face}, UPPER CASE and lower case letters to represent vector/matrix quantities, random variables and sample values respectively.
In order to represent the policy action that corresponds to ``idling'' a server, we introduce a special, ``dummy'' queue which is denoted as queue 0. Allocating a server to this queue is equivalent to idling that server.
By default, queue 0 is permanently connected to all servers and contains only ``dummy'' packets.
Let $\mathds{1}_{\{A\}}$ denote the indicator function for condition $A$.
Throughout this article, we will use the following notation:
\begin{itemize}

  \item $\mathbf{G}(n)$ is an ${(L+1) \times K}$ matrix, where $G_{i,j}(n)$ for $i>0$ is the channel connectivity random variable as defined in Section \ref{YV:intro_section}. By assumption,  $G_{0,j}(n)=1$ for all $j,n$.

  \item $\mathbf{X}(n) = (X_0(n), X_1(n),X_2(n),\ldots, X_L(n))^T$ is the vector of queue lengths at the beginning of time slot $n$, measured in number of packets. We assume $X_0(1)=0$. 

  \item $\mathbf{Y}(n) = (Y_0(n), Y_1(n),Y_2(n),\ldots, Y_L(n))^T$ is the \emph{withdrawal control}.
  For any $i$, $Y_i(n)\in \{0,1,\ldots,K \}$ denotes the number of packets withdrawn from queue $i$ (and assigned to servers) during time slot $n$.

  \item $\mathbf{Z}(n) = (Z_0(n), Z_1(n),Z_2(n),\ldots, Z_L(n))^T$ is the  vector of the number of exogenous arrivals during time slot $n$.  Arrivals to queue $i\neq 0$ are as defined in Section \ref{YV:intro_section}.

  \item 
For ease of reference, we call the tuple $(\mathbf{X}(n),\mathbf{G}(n))$  the ``state'' of the system at the beginning of time slot $n$.
\end{itemize}

For any (feasible) control ($\mathbf{Y}(n)$), the system described previously evolves according to
\begin{equation}\label{sys_evol}
    \mathbf{X}(n+1)=\mathbf{X}(n)- \mathbf{Y}(n) +\mathbf{Z}(n), \qquad n=1,2,\ldots
\end{equation}

We assume that arrivals during time slot $n$ are added after removing served packets. Therefore, packets that arrive during time slot $n$ have no effect on the controller decision at that time slot and may only be withdrawn during $t=n+1$ or later.
For convenience and in order to ensure that $X_0(n)=0$ for all $n$, we define  $Z_0(n) = Y_0(n)$. We define controller policies more formally next.


\subsection{Feasible  Scheduling and Withdrawal Controls}
The withdrawal control defined earlier does not provide any information regarding server allocation. Such information is necessary for our optimality proof. To capture such information, we define the vector
$\mathbf{Q}(n)$, where $Q_j(n)\in \{0,1,\ldots, L \}$ denotes the index of the queue that is selected (according to some rule) to be served by server $j$ during time slot $n$. Note that serving the ``dummy" queue, i.e., setting $Q_j(n)=0$ indicates that server $j$ is idling during time slot $n$. For future reference, we will call $\mathbf{Q}(n)$ the \emph{scheduling (or server allocation) control}.

Using the previous notation and given a scheduling control vector $\mathbf{Q}(n)$ we can compute the withdrawal control vector as:
\begin{equation}\label{Yi_qi}
    Y_i(n) = \sum_{j=1}^K \mathds{1}_{\{ i=Q_j(n)\}}, \quad i=0,1,2,\ldots,L.
\end{equation}

We say that a given vector $\mathbf{Q}(n)\in \{0,1,\ldots, L \}^K$  is a \textsl{feasible} scheduling control (during time slot $n$) if: (a) a server is allocated to a connected queue, and, (b) the number of servers allocated to a queue (dummy queue excluded) cannot exceed the size of the queue at time $n$.
Similarly, we say that a  vector $\mathbf{Y}(n)\in \{0,1,\ldots,K \}^{L+1}$ is a \textsl{feasible} withdrawal control (during time slot $n$) if  there exists a feasible scheduling control $\mathbf{Q}(n)$ that satisfies  Equation  (\ref{Yi_qi}).

Conditions (a) and (b) above are also necessary for feasibility of a scheduling control vector $\mathbf{Q}(n)$. From Equation  (\ref{Yi_qi}), a feasible withdrawal control $\mathbf{Y}(n)$ satisfies the following necessary  conditions:
\begin{eqnarray}
    0  \leq & \hspace{-3mm} Y_i(n) \hspace{-3mm} &\leq \min\left(X_i(n),\sum_{j=1}^K G_{i,j}(n)\right), \, \forall  n,i\neq 0, \label{YVeq:cons1} \\
	 &&  \sum_{i=0}^L Y_i(n)  =  K , \quad \forall   n.  \label{YVeq:cons2}
\end{eqnarray}

For the rest of this article, we will refer to $\mathbf{Q}(n)$ as an \emph{implementation} of the given feasible control $\mathbf{Y}(n)$. We denote the set of all feasible withdrawal controls while in state $(\mathbf{x},\mathbf{g})$ by $\mathcal{Y}(\mathbf{x},\mathbf{g})$.

Note from Equation  (\ref{Yi_qi}) that, given a feasible scheduling control $\mathbf{Q}(n)$, a feasible withdrawal control $\mathbf{Y}(n)$ can be readily constructed.
Note, however, that, for any feasible $\mathbf{Y}(n)$, the feasible scheduling control $\mathbf{Q}(n)$ may not be unique. Furthermore,
given a feasible $ \mathbf{Y}(n)$, the construction of the scheduling control $\mathbf{Q}(n)$  may not be straightforward\footnote{Given a state $(\mathbf{X}(n),\mathbf{G}(n))$ and a feasible withdrawal vector      $\mathbf{Y}(n)$, one can determine the feasible scheduling control by performing a brute-force search over all feasible vectors $\mathbf{Q}(n)$.} and will not be examined in this article.

%

\subsection{Definition of Scheduling Policies} \label{Model Description}

A \emph{scheduling policy} $\pi$ (or policy $\pi$ for simplicity) is a rule that determines feasible withdrawal vectors $\mathbf{Y}(n)$ for all $n$, as a function of the past history and current state of the system $\mathbf{H}(n)$. The state history is given by the sequence of random variables
\begin{eqnarray}
	\mathbf{H}(1)\!\!&=&\!\!(\mathbf{X}(1)), \quad \text{and} \nonumber \\
    \mathbf{H}(n)\!\!\!&=&\!\!\!(\mathbf{X}(1),\mathbf{G}(1),\mathbf{Z}(1),\ldots, \mathbf{G}(n\!-\!1), \mathbf{Z}(n\!-\! 1), \mathbf{G}(n)), \nonumber \\ && \,\, n=2,3,\ldots
\end{eqnarray}

Let $\mathcal{H}_n$ be the set of all state histories up to time slot $n$. Then a policy $\pi$ can be formally defined as the sequence of measurable functions
\begin{eqnarray} \label{eq:policy}
&& u_n: \mathcal{H}_n \longmapsto \mathcal{Z}_+^{L+1}, \nonumber \\
\text{s.t.} &&  u_n(\mathbf{H}(n))\in \mathcal{Y}(\mathbf{X}(n),\mathbf{G}(n)), \quad  n=1,2,\ldots
\end{eqnarray}
where $\mathcal{Z}_+$ is the set of non-negative integers and $\mathcal{Z}_+^{L+1}=\mathcal{Z}_+ \times \cdots \times \mathcal{Z}_+$, where the Cartesian product is taken $L+1$ times.

At each time slot, the following sequence of events happens: first, the connectivities $\mathbf{G}(n)$ and the queue lengths $\mathbf{X}(n)$ are observed. Second, the packet withdrawal vector $\mathbf{Y}(n)$ is determined according to a given policy. Finally, the new arrivals $\mathbf{Z}(n)$ are added to determine the next queue length vector $\mathbf{X}(n+1)$.

We denote the set of all scheduling policies described by Equation (\ref{eq:policy}) by $\Pi$. We introduce next a subset of $\Pi$, namely the class of \emph{Most Balancing} (MB) policies.   The goal of this work is to prove that MB policies are optimal (in a stochastic ordering sense).

\section{The Class of MB Policies}
\label{sec:MBPolicies}
In this section, we provide a  description and mathematical characterization of the class of MB policies.
Intuitively, the MB policies ``attempt to minimize the queue length differences in the system  at every time slot $n$".
For a more formal characterization of MB policies, we first define the following:

Given a state $(\mathbf{x}(n),\mathbf{g}(n))$ and a policy $\pi$ that chooses the feasible control $\mathbf{y}(n) \in \mathcal{Y}(\mathbf x, \mathbf g)$ at time slot $n$, define the ``updated queue size'' $\hat{x}_i(n)= x_i(n)-y_i(n)$ as the size of queue $i, i=0, 1,\ldots, L$,  after applying the control $y_i(n)$ and just before adding the arrivals during time slot $n$. Note that because we let $z_0(n)= y_0(n)$, we have $\hat x_0(n) \in \mathcal Z$ where $\mathcal Z$ is the set of all integers, i.e., we allow $\hat x_0(n)$ to be negative. 

We define  $\kappa_n(\pi)$, the \textit{``imbalance index''} of policy $\pi$  at time slot $n$, as the following sum of differences:  

\begin{equation} \label{eq:kappa}
   \kappa_n(\pi) = \sum_{i=1}^{L} \sum_{j=i+1}^{L+1} (\hat x_{[i]}(n) -\hat x_{[j]}(n) ),
\end{equation}

\noindent where $[k]$ denotes the index of the $k^{th}$ longest queue after applying the control $\mathbf y(n)$ and before adding the arrivals at time slot $n$. By convention, queue `$0$' (the ``dummy queue") will always have order $L+1$ (i.e., the queue with the minimum length). This definition ensures that the differences are nonnegative and a pair of queues is accounted for in the summation only once; moreover, as we shall see in 
Lemma \ref{lem:alg1c}, this definition allows for a straightforward calculation and comparison of various policies\footnote{We have experimented with alternatives to Equation (\ref{eq:kappa})  that use lexicographic ordering of queues and ``Min-Max'' definitions (e.g., minimize the length of the largest queue). However, we were not able to derive results equivalent to Lemma \ref{lem:alg1c}.}. 

It follows from Equation (\ref{eq:kappa}) that the minimum possible value of the imbalance index is equal to $L \cdot \hat x_{[L]}$ (i.e., all $L$ queues have the same length which is equal to the shortest queue length) which is indicative of a fully balanced system. It also follows that the maximum such value is equal to $ 2(L-1) \hat x_{[1]} - (L-2)\hat x_{[L]}$. This value is attained when the $L-1$ longest queues have the same size.

Let $ \Pi^{MB}$ denote the set of all MB policies, then we define the elements of this set as follows:

\noindent\textbf{Definition:}
A \textit{Most Balancing} (MB) policy is a policy $\pi \in \Pi^{MB}$ that, at every $n=1,2, \ldots$, chooses feasible withdrawal vector $\mathbf y(n) \in \mathcal{Y}(\mathbf x, \mathbf g)$ such that the imbalance index  at that time slot is minimized, i.e.,
\begin{equation} \label{eq:MBpolicy}
   \Pi^{MB} = \big\{\pi\in \Pi:\argmin_{\mathbf y(n) \in \mathcal{Y}(\mathbf x, \mathbf g)}  \kappa_n(\pi), \quad \forall n \big \}
\end{equation}

The  set $\Pi^{MB}$ in Equation (\ref{eq:MBpolicy}) is well-defined and non-empty, since the minimization is over a finite set.
Note that the set of MB policies may have more than one element. This could happen, for example, when at a given time slot $n$, a server  $k$ is connected to two or more queues of equal size, which happen to be the longest queues connected to this server after allocating all the other servers. 
To illustrate this case, consider a two-queue system with a single, fully-connected server at time slot $n$. Let $\mathbf x(n) =(5,5)$. Assume that policy $\pi_1$ (respectively $\pi_2$) chooses a withdrawal vector $\mathbf y(n)=(1,0)$ (respectively $\mathbf y^*(n)=(0,1)$). Then both policies minimize the imbalance index, and  $\kappa_n(\pi_1) = \kappa_n(\pi_2) =10$. 

Given $\mathbf{X}(t)$ and $\mathbf{G}(t)$, one can  construct an MB policy using a direct search over all possible server allocations. For large $L$ and $K$, this can be a challenging computational task and is not the focus of this work. In Section \ref {heurestic_implem}, we provide a low-complexity heuristic algorithm (LCSF/LCQ) to approximate MB policies.  

\begin{rem}
Note that the LCQ policy in  \cite{Tassiulas} is a most balancing (MB) policy for $K=1$ (i.e., the one server system presented in \cite{Tassiulas}). Extension of LCQ to $K>1$ (i.e., allocating all the servers to the longest queue in the multiserver model) may not result in a MB policy, as the following example demonstrates.

Consider a system of three queues with three fully-connected servers during time slot $n$. Let $\mathbf x(n)=(6,5,4)$. An LCQ policy in the spirit of \cite{Tassiulas} that allocates all servers to the longest connected queue results in queue size vector $ \mathbf {\hat{x}} (n)=(3,5,4)$. 
Moreover, an LCQ policy in the spirit of \cite{ganti} that allocates the three servers to the three longest connected queues   results in queue size vector $\mathbf {\hat{x}}(n)=(5,4,3)$. Both policies have  $\kappa_n(\pi)=16$.
An MB policy results in queue size vector $\mathbf {\hat{x}}(n)=(4,4,4)$ and $\kappa_n(\pi)=16$. 
\endremark

\end{rem}

\subsection{Comparing arbitrary policies to an MB policy}
When comparing various policies to an MB policy, the definition in Equation (\ref{eq:MBpolicy}) is cumbersome since it involves all time instants $n$.  
The subsets $\Pi_n$ we introduce next define policies that are related to MB policies and  allow us to perform comparisons \textit{one single instant at a time}. 

Consider any fixed $n\geq 1$; we say that a policy $\pi\in\Pi$ ``has the MB property'' at time $n$, if  $\pi$ achieves the minimum value of the index $\kappa_n(\pi)$. 

\textbf{Definition:} For any given time  $n\ge 0$, $\Pi_n$ denotes the set of policies that have the MB property at all  time slots $t\le n$ (and are arbitrary for $t>n$). 

We have that $\Pi=\Pi_0$. Note that the set $\Pi_n$ is not empty, since MB policies are elements of it. We can easily see that these sets form a monotone sequence, with 

\begin{equation}\label{Pin-monotone}
\Pi_n \subseteq  \Pi_{n-1}
\end{equation}

\noindent Then the set $\Pi^{MB}$ in Equation (\ref{eq:MBpolicy}) can be defined as 

\[
\Pi^{MB}= \bigcap_{n=1}^{\infty} \Pi_n.
\]


The vector $\mathbf D$ defined in Equation (\ref{eq:I1}) is a measure of how much an arbitrary policy $\pi$ differs from a given MB policy during a given time slot $n$.

\textbf{Definition:} Consider a given state $(\mathbf x(n), \mathbf g(n))$  and a policy $\pi $ that chooses the feasible withdrawal vector $\mathbf y(n)$ during time slot $n$. Let $\mathbf y^{MB}(n)$ be a withdrawal vector chosen by an MB policy during the same time slot $n$. We define the ($L+1)\times 1$-dimensional vector $\mathbf D \in \mathcal Z^{L+1}$ as
\begin{equation}\label{eq:I1}
	\mathbf D= \mathbf y^{MB}(n) - \mathbf y(n).
\end{equation}

\noindent Note that, for notational simplicity, we omit the dependence of $\mathbf D$ on the policies and the time index $n$. Intuitively, a negative element $D_i$ of vector  $\mathbf D$  indicates that more packets than necessary (compared to a policy that has the MB property) have been removed from queue $i$ under policy $\pi$. 
 
The following lemma quantifies the difference between an arbitrary policy and an MB policy (at time $n$). Its proof is given in Appendix \ref{appendixB1iii}. 

\begin{lem} \label{lem:I1}
Consider a given state $(\mathbf x(n), \mathbf g(n))$ and a policy $\pi\in \Pi$. Then, (a) if $\mathbf D=0$, the policy $\pi$ has the MB property at time $n$, and, (b) if $\pi$ has the MB property at time $n$, the vector $\mathbf D$ has components that are $0, +1, $ or $-1$ only.
\end{lem}

Consider a policy $\pi\in \Pi$; let $h_{\pi} = \sum_{i=0}^L |D_i|/2 $.  As we show in the Appendix (see Lemma \ref{lem:I2}), $h_{\pi} $ is integer-valued and $0 \le h_{\pi}  \le K$.  In view of Lemma \ref{lem:I1}, $h_{\pi} $ can be seen as a measure of ``how close'' the policy $\pi$ is to having the MB property at time $n$.   

\textbf{Definition:} For any given time  $n$ and integer $ h$, where $0 \leq h \leq K $, define the set $\Pi_n^h $ as the set that contains all policies $\pi \in \Pi_{n-1}$, such that $ h_{\pi}  \leq h$. 


From Lemma \ref{lem:I1}, we can see that $\Pi_n^0 = \Pi_n$. 
We can  easily check that $ \Pi_n^K = \Pi_{n-1}$, so $\pi \in \Pi_n^K$ by default. 
$\{\Pi_n^{h}  \}_{h=0}^K$ forms a monotone sequence, with   

\begin{equation}\label{Pinh-sequence}
\Pi_n = \Pi_n^0 \subseteq \cdots \subseteq  \Pi_n^{h}  \subseteq \cdots \subseteq\Pi_n^K
= \Pi_{n-1}
\end{equation} 

\noindent We exploit the monotonicity property of the $\Pi_n^{h}$ sets in the next section, when we show how balancing interchanges reduce the imbalance index of a given policy. 

 Note that the set $\Pi$ of all policies can be denoted as 

\begin{equation}\label{Pinh-sequence1}
 \Pi = \Pi_0 =\cup_{h=0}^K \Pi_1^h.
\end{equation} 
 
It follows from the last two equations that an arbitrary policy $\pi\in\Pi$ will also belong to a set $\Pi_{n-1}$, for some $n\ge 1$. The proof of optimality in Section \ref{YV:main-result} is based on comparisons of $\pi$ to a series of policies that belong to the subsets $\Pi_n^{h}$ (see Lemma \ref{lem:6}). 

\section{Balancing Interchanges}\label{YV:Balancing-Interchanges}

In this section, we introduce the notion of ``balancing interchanges''. Intuitively, an \emph{interchange} $\mathbf I(f,t)$ between two queues, $f$ and $t$, describes the action of withdrawing a packet from queue $f$ instead of queue $t$ (see Equations (\ref{eq:30002}) and (\ref{eq:30003})).  Such interchanges are used to relate the imbalance indices of various policies (see Equation (\ref{eq:I001})); \emph{balancing} interchanges are special in two ways: (a) they do not increase the imbalance index (see Lemma \ref{lem:YV1}) and thus provide a means to describe how a policy can be modified to obtain the MB property at time $n$, and, (b) they preserve the queue size ordering we define in the next section (see relations R1-R3 in Section \ref{PreferredOrder-section}). This ordering is crucial in proving optimality.

Interchanges can be implemented via server reallocation. Since there are $K$ servers, it is intuitive that at most $K$ interchanges suffice to convert any arbitrary policy to a policy that has the MB property at time $n$. The crux of Lemma \ref{lem:I6}, the main result of this section, is that such interchanges are \textit{balancing}.  

\subsection{Interchanges between two queues}
Let $f \in \{0,1,\dots,L\}, ~t \in \{0,1,\dots,L\}$ represent the indices of two queues that we refer to as the `from' and `to' queues. Define the $(L+1)\times 1$-dimensional  vector $\mathbf I(f,t)$, whose $j$-th element is given by:
\begin{eqnarray}\label{eq:newI4}
I_j(f,t) &=&
 \left\{
  \begin{array}{cc}
    0, & \hbox{$ t = f $;} \\
    +1, & \hbox{$ j = f,  ~ f\ne t$;} \\
    -1, & \hbox{$ j = t,  ~t\ne f$;} \\
    0, & \hbox{otherwise.}
  \end{array}
\right.
\end{eqnarray}


Fix an initial state $(\mathbf{x}(n),\mathbf{g}(n)) $ at time slot $n$; consider a policy $\pi$ with a (feasible) withdrawal vector $\mathbf{y}(n)$. Let
\begin{equation}\label{eq:I30000}
 \mathbf{y^*}(n) = \mathbf{y}(n) + \mathbf I(f,t), ~~f\ne t,
\end{equation}
be another withdrawal vector. The two vectors $\mathbf{y}(n),\mathbf{y^*}(n)$ differ only in the two components $t,f$; under the withdrawal vector $\mathbf{y^*}(n)$, an additional packet is removed from queue $f$, while one  packet less is removed from queue $t$. Note that either $t$ or $f$ can be the dummy queue. In other words,
\begin{eqnarray}\label{eq:30002}
  y^*_f(n) &=& y_f(n) + 1  \\
  y^*_t(n) &=& y_t(n) - 1 \label{eq:30003}  \\
  y^*_i(n) &=& y_i(n), \quad \forall i\ne f,t. \label{eq:30004}
\end{eqnarray}

In the sequel, we will call $\mathbf I(f,t)$ an \emph{interchange} between queues $f$ and $t$.   We will call $\mathbf I(f,t)$ a \emph{feasible interchange} if it results in a feasible withdrawal vector $\mathbf{y^*}(n)$.
It follows immediately from Equations (\ref{sys_evol})  
and (\ref{eq:I30000}) that the $\mathbf I(f,t)$ interchange will result in a new vector, $\mathbf{\hat x}^*(n)$, of updated queue sizes, such that:
%
%
\begin{equation}\label{eq:I30000yv}
 \mathbf{\hat x^*}(n) = \mathbf{\hat x}(n) - \mathbf I(f,t), ~~f\ne t.
\end{equation}

We are interested next in describing sufficient conditions for ensuring feasible interchanges.
 
\subsection{Feasible Single-Server Reallocation}

Given the state $(\mathbf{x}(n),\mathbf{g}(n)) $, let $\mathbf y(n)$ be any feasible withdrawal vector at time slot $n$ that is implemented via $\mathbf q(n)$. We define a ``feasible, single-server reallocation'' (from queue $t$ to queue $f$) as the reallocation of a single server $k$ from queue $t$ to queue $f$, such that the new scheduling control $\mathbf q^*(n)$ is also \textit{feasible}.  The conditions $g_{f,k}(n)\cdot g_{t,k}(n)\cdot \mathds 1_{\{ q_{k}(n) = t \}} =1$ and $\hat x_f(n) \geq 1$ are sufficient for the reallocation of server $k$  (from queue $t$ to queue $f$) to be feasible. 

A feasible, single-server  reallocation from queue $t$ to queue $f$ results into a feasible interchange $\mathbf I(f,t)$. However, the reverse may not be true, as we detail in the following section. 

\subsection{Sufficient conditions  for a feasible  interchange}

Consider again the state $(\mathbf{x}(n),\mathbf{g}(n))$ and feasible scheduling control $\mathbf q(n)$. 
The feasible interchange $\mathbf I(f,t)$ in Equation (\ref{eq:I30000})  may result from  a \textit{sequence} of $m$ feasible, single-server reallocations among several queues, as demonstrated in Figure \ref{example}, where $1 \leq m \leq K$. 

\begin{figure}
\centering
\includegraphics[width=2.7in]{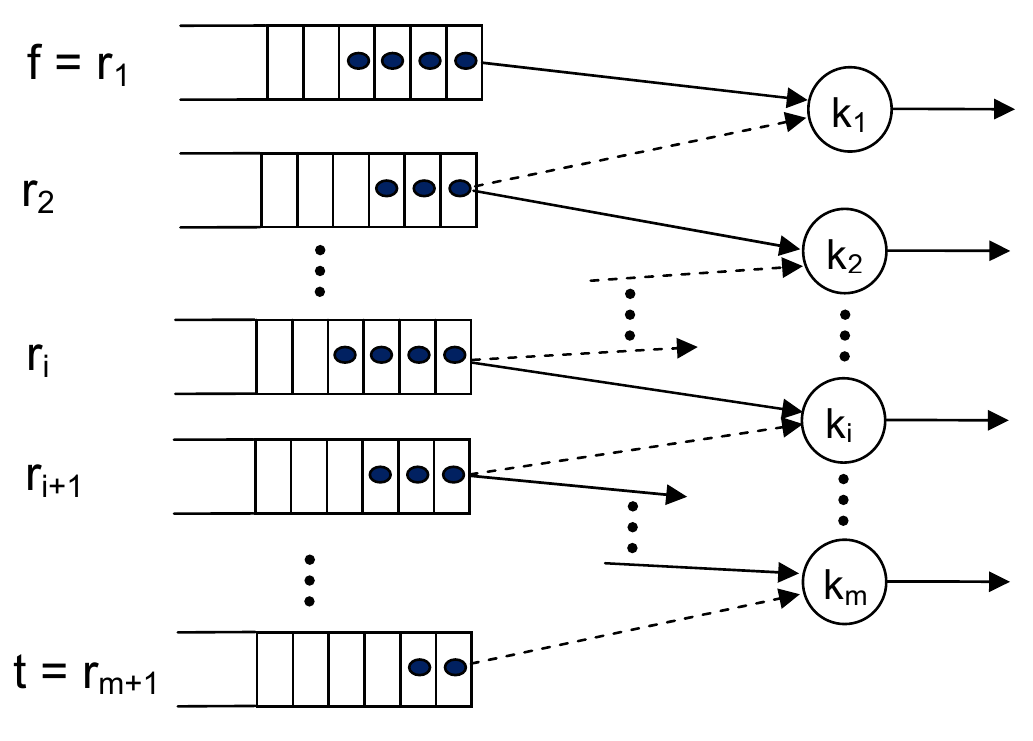}
\caption{A sequence of $m$ single-server reallocations results in a feasible  $\mathbf I(f,t)$. The dotted line denotes original server allocation. The solid line denotes server reallocation that implements $\mathbf I(f,t)$.}
\label{example}
\end{figure}

Let $\mathbf r \in \{0,1, \ldots, L\}^{m+1}$ denote a sequence of queue indices, where $r_1 = f$  and $r_{m+1} =t$. Let $k_i: k_i \in \{1,2,\ldots, K\}$ denote the server reallocated from queue $r_{i+1}$ to queue $r_i$.
Then the following are  sufficient  conditions for the feasibility of the interchange operation of Equation (\ref{eq:I30000}):

\begin{eqnarray}\label{eq:I101}
 &&\sum_{i=1}^{m} g_{r_i,k_i}(n)\cdot g_{r_{i+1},k_i}(n)\cdot \mathds 1_{\{ q_{k_i}(n) = r_{i+1} \}} = m,  \\
 &&\hat x_f(n) \ge 1, \quad \text{if } \, f \in \{1,2,\ldots,L \},    \label{eq:I102.1} 
\end{eqnarray}
for some integer $1\le m\le K$ and $\mathbf r \in  \{0,1, \ldots, L\}^{m+1}$.

Constraint (\ref{eq:I101}) ensures  that connectivity conditions allow for the feasibility of all $m$ intermediate single-server reallocations. The sequence of  server reallocations
starts by  reallocating  server $k_1$ to queue $f= r_1$. In this case,  queue $r_1 $ is reduced by one packet (i.e., an extra packet is withdrawn from queue $f$) and  queue $r_2$ is increased by one packet.  Constraint (\ref{eq:I102.1}) ensures that a
packet can be withdrawn from queue $f$. The reallocation of server $k_1$ insures
that queue $r_2$  contains at least one packet for the second
intermediate single-server reallocation to be feasible even when
$\hat x_{r_2} (n) = 0$. Same is true for any queue $r_i: i \in \{2, 3, \ldots, m\}$.
Therefore,  constraints (\ref{eq:I101}) and (\ref{eq:I102.1}) are also sufficient for the feasibility
of the interchange $\mathbf I(f,t)$.




\subsection{``Balancing'' interchanges} \label{sec:BI}

\textbf{Definition:} A feasible interchange $\mathbf I(f,t)$ is  ``\textit{balancing}'' if
\begin{equation}\label{eq:I001YV}
\hat x_f(n) \geq \hat x_t(n) +1.
\end{equation}

$\mathbf I(f,t)$ is     ``\textit{unbalancing}'' if
\begin{equation}\label{eq:I002YV}
\hat x_f(n) \le \hat x_t(n).
\end{equation}

Balancing interchanges result in policies that may reduce the imbalance index, as the following lemma states.

\begin{lem} \label{lem:YV1}
Consider two policies $\pi^*$ and $\pi$, related via the balancing interchange
\[
\mathbf y^*(n)= \mathbf y(n) + \mathbf I(f,t),
\]
at time slot $n$. Then the imbalance indices for the two policies  are related via
\begin{equation}\label{eq:I001}
	\kappa_n (\pi^*)= \kappa_n (\pi) - 2(s-l) \cdot \mathds 1_{\{ \hat x_{[l]}(n) \geq \hat x_{[s]}(n) +2 \} }
\end{equation}
\noindent where $l$ (respectively $s$) is the order of queue $f$ (respectively $t$) in $\hat x(n)$ when ordered in descending order, such that, $ s>l; x_l> x_a, \forall a>l$ and $x_s < x_b, \forall b<s$.\footnote{These conditions state that when there exist multiple components that have the same value as $x_l$ (respectively $x_s$) only the last  (respectively the first) of the components in order is considered. Intuitively, we use $s$ (respectively $l$) to refer to the order of the ``shorter''  (respectively the ``longer'')  queue of the two queues used in the interchange.}.
\end{lem}

The proof is a direct consequence of Lemma \ref{lem:alg1c} in Appendix \ref{appendixB1} and the fact that, by definition  of the balancing interchange, we have $s>l$.

In words, Equation (\ref{eq:I001}) states that an interchange $\mathbf I(f,t)$, when balancing, results in: either a cost reduction of $2(s-l)$  (when $\hat x_{f}(n)=\hat x_{[l]}(n) \geq \hat x_{[s]}(n) +2 = \hat x_{t}(n)+2$) or an unchanged cost (when  $\hat x_{f}(n)= \hat x_{t}(n)+1$). The latter case agrees with intuition, since the balancing interchange in this case will result in simply permuting the lengths of queues $f$ and $t$; this permutation does not change the total sum of differences (and hence the imbalance index) in the resulting queue length vector.  

We determine next conditions that characterize what interchanges are balancing. We also  describe how balancing interchanges transform an arbitrary policy to an MB policy. 

\subsection{How to determine balancing interchanges}\label{section-lem:I4}




Lemma \ref{lem:I4} provides a selection criterion to systematically select balancing (and hence improving) interchanges. Lemma \ref{lem:I6} provides a bound on the number of  interchanges needed   to convert any policy into one that has the  MB property at time $n$. The proofs of the two lemmas are given in Appendix \ref{appendixB1i} and \ref{appendixC1i} respectively.

\begin{lem} \label{lem:I4}
Consider a given state $(\mathbf x(n), \mathbf g(n))$ and a feasible withdrawal vector $\mathbf y(n)$. Any feasible interchange $\mathbf I(f,t)$ with indices $f$ and $t$ such that $D_f \geq +1$, $D_t \leq -1$  is a balancing interchange.
\end{lem}



Recall that $h_{\pi}  =\sum_{i=0}^L |D_i|/2$. Consider a sequence of balancing interchanges, $\mathbf I(f_1,t_1), \mathbf I(f_2,t_2),\dots, \mathbf I(f_{h_{\pi}},t_{h_{\pi}})$. Let 

\[
\mathbf {y}^*(n) = \mathbf {y}(n) + \sum_{i=1}^{ h_{\pi}} \mathbf I(f_i,t_i).
\]
We denote by $\pi^* $ the policy that chooses the withdrawal vector $\mathbf {y}^*(n)$. In other  words, $\pi^* $ denotes the policy that results from applying this sequence of  interchanges.

\begin{lem} \label{lem:I6}
For any policy $\pi \in \Pi_{n-1}$,  $h_{\pi}  $ balancing interchanges 
suffice to determine a  policy $\pi^* $ such that $\pi^* \in \Pi_n$.
\end{lem}

Lemma \ref{lem:I4} can be used to identify queues $f_i$ and $t_i$ during time slot $n$ such that the interchange $\mathbf I(f_i,t_i)$ is balancing. Lemma  \ref{lem:I6} shows that performing a sequence of such interchanges, determines a policy that has the MB property for one more time slot. Both lemmas are crucial  for the proof of our main result, since they indicate how a given policy can be improved  using one balancing interchange at a time.

\section{Optimality of MB Policies}
\label{YV:main-result}

In this section, we present the main result of this article, that is, the optimality of the Most Balancing (MB) policies. We will establish  optimality for a range of performance criteria, including the minimization of the total number of packets in the system. We introduce the following definition.

\subsection{Definition of Preferred Order}
\label{PreferredOrder-section}
\noindent Let's define the relation $\preceq$ on $\mathcal{Z}_+^{(L+1)}$ first;
we say $\mathbf{\tilde{x}} \,\preceq  \, \mathbf{x}$ if:
\begin{enumerate}
\item[R1-]    $\tilde{x}_i \leq x_i$    for all $i$ (i.e., point wise comparison), 
\item[R2-] $\mathbf{\tilde{x}}$ is obtained from $\mathbf{x}$ by permuting two of its components; the two vectors differ only in two components $i$ and $j$, such that $\tilde{x}_i = x_j$ and $\tilde{x}_j = x_i$, or
\item[R3-]	$\mathbf{\tilde{x}}$ is obtained from $\mathbf{x}$ by performing a \textsl{``balancing interchange"}, in the sense of Equation (\ref{eq:I001YV}), i.e.,
the two vectors differ in two components  $i>0$ and $j\ge 0$ only, where $x_i \geq x_j +1$,
such that: $ \tilde{x}_i = x_i-1$ and $\tilde{x}_j = x_j+1$.
\end{enumerate}

To prove the optimality of MB policies, we will need a methodology that enables comparison of the queue lengths under different policies. Towards this end,  we define a ``preferred order'' as follows:

\noindent \textbf{Definition: (Preferred Order)}. The transitive closure of the relation $\preceq$ defines a partial order (which we call \textit{preferred order} and use the symbol $\prec_p$ to represent) on  the set $\mathcal{Z}_+^{(L+1)}$.~\endremark

The transitive closure \cite{Lidl}, \cite{ganti} of $\preceq$ on the set $\mathcal{Z}_+^{(L+1)}$ is the smallest transitive relation on $\mathcal{Z}_+^{(L+1)}$ that contains the relation $\preceq$. From the engineering point of view,
$\mathbf{\tilde{x}} \prec_p \mathbf{x}$  if  $\mathbf{\tilde{x}} $ is obtained from $\mathbf{x}$ by performing a sequence of \textit{reductions}, \textit{permutations of two components} and/or \textit{balancing interchanges}.

For example, if $\mathbf{\tilde{x}}  = (3,4,5)$ and $\mathbf{x}  = (4,5,3)$ then $\mathbf{\tilde{x}} \prec_p \mathbf{x}$ since $\mathbf{\tilde{x}} $ can be obtained from $ \mathbf{x}$ by performing the following two consecutive two-component permutations: first swap the second and third components of $ \mathbf{x}$, yielding $ \mathbf{x}^1 = (4,3,5)$ then swap the first and  second components of $ \mathbf{x}^1$, yielding $ \mathbf{x}^2 = (3,4,5) = \mathbf{\tilde{x}} $.

Suppose that $\mathbf{\tilde{x}}, \mathbf{x}$ represent queue size vectors for our model.
Statement R3 in this case describes moving a packet from one real, large queue $i$ to another  smaller one $j$ (note that the queue with index $ j = 0$  is not excluded since a balancing interchange may represent the allocation of an idled server).
We say that $\mathbf{\tilde{x}}$ is \emph{more balanced} than $\mathbf{x}$ when R3 is satisfied. For example, if $L=2$ and $\mathbf{x}=(0,5,2)$ then a balancing interchange (where $i=1$ and $j=2$) will result in $\mathbf{\tilde{x}}=(0,4,3)$.

\subsection{The class $\mathcal{F}$ of cost functions }\label{YV:cost-functions}
Let $\mathbf{\tilde{x}}, \mathbf{x} \in \mathcal{Z}_+^{(L+1)}$ be two vectors representing queue lengths. Then we denote by $\mathcal{F}$ the class of real-valued functions on $\mathcal{Z}_+^{(L+1)}$ that are monotone, non-decreasing with respect to the partial order $\prec_p$; that is, $f \in \mathcal{F}$ if and only if
\begin{equation} \label{func_class}
    \mathbf{\tilde{x}} \prec_p \mathbf{x} \, \Rightarrow f(\mathbf{\tilde{x}}) \leq f(\mathbf{x}).
\end{equation}

From  (\ref{func_class}) and the definition of preferred order, it can be easily seen that the function $f(\mathbf{x})=x_1+x_2+\cdots +x_L$ belongs to $\mathcal{F}$. This function corresponds to the total number of queued packets in the system\footnote{Another example is the function $f'(\mathbf{x})=\max\{x_1, \ldots ,x_L\}$ which also belongs to the class $\mathcal{F}$.}.





For two real-valued random variables $A$ and $B$, $A \leq_{st} B$ defines the usual stochastic ordering \cite{Stoyan}. In the remainder of this paper, we say that a policy $\sigma$ \textit{dominates} another policy $\pi$ if

\begin{equation} \label{policy_dominance}
 f(\mathbf{X}^{\sigma}(t)) \leq_{st}  f(\mathbf{X}^{\pi}(t)) , \quad \forall~t=1,2,\dots
\end{equation}

\noindent for all cost functions $f \in \mathcal{F}$.

We will need the following lemma to complete the proof of our main result presented in Theorem \ref{thm:3}.

\begin{lem}\label{lem:6}
Consider an arbitrary policy $\pi \in \Pi_{\tau}^h$, where $h>0$. Then, there exists a policy $\tilde{\pi} \in \Pi_{\tau}^{h-1}$, such that $\tilde{\pi} $ dominates $\pi$.
\end{lem}

The full details of the proof for Lemma \ref{lem:6} are given in Appendix \ref{appendixB}.
The proof involves two parts. First, we construct a policy $\tilde{\pi} $ by applying a balancing interchange to $\pi$; using Lemmas \ref{lem:I4} and \ref{lem:I6}, we show that $\tilde{\pi} \in \Pi_{\tau}^{h-1}$. 
Second, we prove that $\tilde{\pi} $  dominates  policy $\pi$ (see Equation (\ref{policy_dominance})); this part employs coupling arguments.

\subsection{The main result}\label{sec:main-result}

In the following, $\mathbf{X}^{MB}$ and $\mathbf{X}^{\pi}$ represent the queue sizes under a MB and an arbitrary policy $\pi$. 

\begin{thm}\label{thm:3}
Consider a system of $L$ queues served by $K$ identical servers, as shown in Figure \ref{fig_1} with the assumptions of Section \ref{YV:intro_section}. A Most Balancing (MB) policy dominates any arbitrary policy when applied to this system, i.e.,
\begin{equation} \label{dominance_theorem}
 f(\mathbf{X}^{MB}(t)) \leq_{st}  f(\mathbf{X}^{\pi}(t)) , \quad \forall~t=1,2,\dots
\end{equation}
\noindent for all $\pi\in \Pi$ and all cost functions $f \in \mathcal{F}$.
\end{thm}

\begin{proof}  
From (\ref{func_class}) and the definition of stochastic dominance,
it is sufficient to show that $\mathbf{X}^{MB}(t) \prec_p \mathbf{X}^{\pi}(t)$ for all $t$ and all sample paths in a suitable sample space. The sample space is the standard one used in  stochastic coupling methods \cite{Lindvall}; see Appendix \ref{appendixB} for more details.

To prove the optimality of an MB policy, $\pi^{MB}$, we start with an arbitrary policy $\pi$ and apply a series of modifications that result in a sequence of policies ($\pi_1,\pi_2,\ldots$). The modified policies have the following properties:
\begin{enumerate}
\item[(a)] $\pi_1 $ dominates the given policy $\pi$,
\item[(b)] $\pi_i \in \Pi_{i}$, i.e.,  policy $\pi_i $ has the MB property at time slots $t=1,2,\ldots,i$, and,
\item[(c)] $\pi_j $ dominates $\pi_i $ for $j>i$ (i.e., $\pi_j$ has the MB property for a longer period of time than $\pi_i$).
\end{enumerate}

Let $\pi$ be any arbitrary policy; then $\pi \in \Pi_0 = \Pi_1^K$.
Using Lemma \ref{lem:6} we can construct a policy $\tilde{\pi} \in  \Pi_1^{K-1}$ that dominates the original policy $\pi$. Repeating this operation we can construct policies that belong to $\Pi_1^{K-1}, \Pi_1^{K-2}, \ldots, \Pi_1^0 = \Pi_1$ such that all dominate the original policy $\pi$. This sequence of construction steps will result in a policy $\pi_1$ that has the MB property at $t=1$, i.e., $\pi_1 \in \Pi_1 $, and dominates $\pi$. Therefore, by construction $\pi_1 \in \Pi_2^K $. We repeat the construction steps above for time slot $t=2$, by improving on $\pi_1$, to obtain a policy $\pi_2 \in \Pi_2$ that dominates  $\pi_1$, and recursively for $t=3,4,\ldots$ to obtain policies $\pi_3 , \pi_4  , \ldots$ .
%
From the construction of $\pi_n, n=1,2,\ldots$, we can see that it satisfies properties (a), (b) and (c) above.

Denote the limiting policy as $n\longrightarrow \infty$ by $\pi^*$. One can see that $\pi^*$ is an MB policy. Furthermore, $\pi^*$ dominates $\pi_i$, for all $i < \infty$, as well as the original policy $\pi$.
%
\end{proof}\vspace{3mm}

\begin{rem} The optimal policy may not be unique.  Our  main objective is to  prove the optimality of the MB policy not its uniqueness.
The optimality of MB policies makes intuitive sense;  any such policy will tend to reduce the chance that any server idles. This is because an MB policy distributes the servers among the connected queues in the system such that it  keeps packets spread  among all the queues in a ``uniform'' manner.
\endremark
\end{rem}

\section{The Least Balancing Policies}
The \emph{Least Balancing} (LB) policies are the scheduling policies, among all  work-conserving (non-idling) policies, that at every time slot ($n=1,2,\ldots$), choose a packet withdrawal vector $\mathbf{y}(n) \in \mathcal{Y}(\mathbf{x},\mathbf{g})$ that ``maximizes the  differences'' between queue lengths in the system (i.e., maximizes $\kappa_n(\pi)$ in Equation (\ref{eq:kappa})).
In other words, if $\Pi^{LB}$ is the set of all LB policies and $\Pi^{WC}$ is the set of all work conserving policies then,
\begin{equation} \label{eq:LB}
 \Pi^{LB} = \big\{\pi:\argmax_{\mathbf y(n) \in \mathcal{Y}(\mathbf x, \mathbf g)}  \kappa_n(\pi), \pi \in \Pi^{WC}, \quad \forall n \big \}
\end{equation}

Maximizing the imbalance among the queues in the system will result in
maximizing the number of empty queues at any time slot, thus maximizing the chance that servers are forced to idle in future time slots. This intuitively suggests that LB policies will be outperformed by any work conserving policy.
The next theorem  states this fact. Its proof  is analogous to that of Theorem \ref{thm:3} and will not be given here.

\begin{rem}
A non-work conserving policy can by constructed such that it will perform worse than  LB policies, e.g., a policy that  idles all servers.
\endremark
\end{rem}


\begin{thm}\label{thm:4}
Consider a system of $L$ queues served by $K$ identical servers, under the assumptions described in Sections \ref{YV:intro_section}. A Least Balancing (LB) policy is dominated by any arbitrary work conserving policy when applied to this system, i.e.,
\begin{equation} \label{LB-dominance_theorem}
 f(\mathbf{X}^{\pi}(t)) \leq_{st}  f(\mathbf{X}^{LB}(t)) , \quad \forall~t=1,2,\dots
\end{equation}
for all $\pi\in \Pi^{WC}$ and all cost functions $f \in \mathcal{F}$.
\end{thm}


An LB policy has no practical significance, since it maximizes the cost functions presented earlier.  Intuitively, it should also worsen the system stability region and hence the system throughput. However, it is interesting to study the worst possible policy behavior and to measure its performance. The LB and MB policies provide lower and upper limits to the performance of any work conserving policy. The performance of any policy can be measured by the deviation of its behavior from that of the MB and LB policies.

\section{Heuristic Implementation Algorithms For MB and LB Policies} \label{heurestic_implem}

In this section, we present two heuristic policies that approximate the behavior of the MB and LB policies respectively. We present an implementation algorithm for each  one of them.

\subsection{Approximate Implementation  of MB Policies} \label{LCSFLCQ}
We introduce the \textsl{Least Connected Server First/Longest Connected Queue} (LCSF/LCQ) policy, a low-overhead approximation of MB policy, with $O(L\times K)$   computational complexity. The policy is stationary and depends only on the current state $(\mathbf{X}(n),\mathbf{G}(n))$ during time slot $n$.

The LCSF/LCQ implementation during a given time slot is described as follows: The least connected server is identified and is allocated to its longest connected queue. The queue length is updated (i.e., decremented). We proceed accordingly to the next least connected server until all servers are assigned.
In algorithmic terms, the LCSF/LCQ policy can be described/implemented as follows:

\noindent Let  $\mathds{Q}_j=\{i: i=1,2,\ldots,L; g_{i,j}(t)=1 \}$ denote the set of queues that are connected to server $j$ during time slot $t$; we omit the  dependence on $t$ to simplify notation.
Let $\mathds{Q}_{[i]}$ be the $i^{th}$ element in the sequence $(\mathds{Q}_1, \dots, \mathds{Q}_K)$, when ordered in ascending manner according to their size (set cardinality), i.e., $|\mathds{Q}_{[l]}| \geq |\mathds{Q}_{[m]}|$ if $l>m$. Ties are broken arbitrarily. Then under the LCSF/LCQ policy, the $K$ servers are allocated according to the following algorithm:

\begin{alg}[LCSF/LCQ Implementation]\label{eq3}
\begin{flalign*}
	1.& \hspace{2mm} for \quad t  = 1,2,\ldots \quad do \quad  \Big\{ \nonumber\\
	2.& \hspace{6mm} \text{Input:} \,\, \mathbf{X}(t),\mathbf{G}(t).\,\, \text{Calculate} \quad \mathds{Q}_{[l]},\, l = 1,\dots,K.   \nonumber\\
    3.& \hspace{6mm} \mathbf{X}'   \longleftarrow  \mathbf{X}(t), \,\,\mathbf{Y} \longleftarrow  \mathbf{0}, \,\,\mathbf{Q} \longleftarrow  \mathbf{0}\nonumber\\
	4.& \hspace{6mm}  for \,\,  j \,= \,  1 \quad to \quad K  \quad \Big\{\,\,  \text{; allocate servers sequentially} \nonumber\\
	5.& \hspace{9.0mm}	 Q_{[j]} =  \min \left(  \,l:l\in                                                                                                                                                                                                                                                                                                                                                           \!\left\lbrace \argmax_{k:k\in\mathds{Q}_{[j]}} (X'_k|X'_k>0)\right\rbrace  \right)   \nonumber\\	
	6.& \hspace{9mm}	for \,\, i \,=\,  1 \quad to \quad L \quad \Big\{ \nonumber\\
	7.& \hspace{12mm}	 Y_{i} = Y_{i}+ \mathds{1}_{ \left\lbrace i= Q_{[j]} \right\rbrace  } \nonumber\\
	8.& \hspace{12mm}		 X'_i =X_i(t)-Y_i \quad \Big\} \quad \Big\}  \nonumber \\
	9.&  \hspace{6mm} \text{Output:} \,\,\mathbf{y}(t) \longleftarrow \mathbf{Y}, \mathbf{q}(t) \longleftarrow \mathbf{Q}\quad \text{; report outputs}  \nonumber\\
	10.&  \hspace{6mm} \Big\} \qquad \text{; End of Algorithm \ref{eq3}.}  \nonumber
\end{flalign*}
\end{alg}

Note that in line 5 of Algorithm \ref{eq3}, if the set $\mathds{Q}_{[j]}$ is empty, then the $\argmax$ returns the empty set. In this case, the $j^{th}$ order server will not be allocated (i.e., will be idle during time slot $t$). Algorithm \ref{eq3} produces two outputs, when it is run at $t=n$: $\mathbf{y}(n) $ and $\mathbf{q}(n)$ as shown in line 9 of the algorithm. In accordance to the definition of a policy in Equation (\ref{eq:policy}), the LCSF/LCQ policy can be formally defined as the sequence of time-independent mappings $u(\mathbf{x}(n),\mathbf{g}(n) )$ that produce the withdrawal vector $\mathbf{y}(n)$ described in line 9 above.

\begin{lem}\label{lem:LCSF_not_MB}
LCSF/LCQ is not an MB policy.
\end{lem}

To prove lemma \ref{lem:LCSF_not_MB} we present the following counter example. Consider a system with $L=4$ and $K=7$. At time slot $n$ the system has the following
configuration:

The queue state at time slot $n$ is $\mathbf x(n) = (5,5,5,4)$. Servers 1 to 6 are connected to  queues 1, 2 and 3 and server 7 is connected to queues 1 and 4 only.

Under this configuration, we can show that the LCSF/LCQ algorithm will result in $\mathbf{\hat x}(n)= (0,2,3,3,4)$ (where the first element represents the dummy queue that by assumption holds no real packets) and $  \kappa_n(LCSF/LCQ) = 18$. A policy $\pi$ can be constructed that selects the feasible server allocation $\mathbf q=(1,2,3,1,2,3,4)$ which yields the state $\mathbf{\hat x}(n)= (0,3,3,3,3)$ and $\kappa_n(\pi) = 12 < \kappa_n(LCSF/LCQ)$.   Therefore,  LCSF/LCQ is not an MB policy.

The LCSF/LCQ policy is of particular interest for the following reasons: (a) It follows a particular server allocation ordering (LCSF) to their longest connected queues (LCQ) and thus it can be implemented using simple sequential server allocation with low computation complexity,
(b) the selected server ordering (LCSF) and allocation (LCQ) intuitively attempt to reduce the size of the longest connected queue thus reducing the imbalance among queues, and, (c) as  we will see in Section \ref{sec:SimRes}, the LCSF/LCQ performance is statistically indistinguishable from that of an MB policy (implying that the counterexamples similar to the one in Lemma \ref{lem:LCSF_not_MB} proof have low probability of occurrence under LCSF/LCQ system operation). Therefore, LCSF/LCQ can be proposed as an approximate heuristic for the implementation of MB policies.

\subsection{Approximate Implementation of LB Policies} \label{MCSFSCQ}
In this section, we present the MCSF/SCQ policy as a low complexity approximation of LB policies. We also provide an implementation algorithm for MCSF/SCQ using the same sequential server allocation principle that we used in  Algorithm \ref{eq3} above.

The \textsl{Most Connected Server First/Shortest Connected Queue} (MCSF/SCQ) policy is the server allocation policy that allocates each one of the $K$ servers to its shortest connected queue (not counting the packets already scheduled for service) starting with the most connected server first.
The MCSF/SCQ implementation algorithm is analogous to Algorithm \ref{eq3} except for lines 4 and 5 which are  described next:

\begin{alg}[MCSF/SCQ Implementation]\label{eq:alg2}
\begin{flalign*}
	1.&  \hspace{2mm} for \quad t = 1,2,\ldots \quad do \quad   \nonumber\\
\hspace{5mm}\vdots \nonumber\\
	4.& \hspace{5mm}  for \,\,  j  =    K \quad to \quad 1  \quad \,\, \text{; Servers in descending order} \nonumber\\
	5.& \hspace{9mm}	 Q_{[j]} =  \min \left(  \,l:l\in                                                                                                                                                                                                                                                                                                                                                           \!\left\lbrace \argmin_{k:k\in\mathds{Q}_{[j]}} (X'_k|X'_k>0)\right\rbrace  \right)   \nonumber\\	
\hspace{5mm}\vdots \nonumber\\
	10.& \hspace{5mm}   \text{; End of Algorithm \ref{eq:alg2}.} \nonumber
\end{flalign*}
\end{alg}

Comments analogous to the ones valid for Algorithm \ref{eq3} are also valid for Algorithm \ref{eq:alg2}.

\section{Performance Evaluation and Simulation Results} \label{sec:SimRes}

We used simulation to study the performance of the system under the MB/LB  policies and to compare against the system performance under several other policies.
The  metric we used in this study is $EQ \triangleq E(\sum_{i=1}^L X_i)$, the  average of the total number of packets in the system.

We focused on two groups of simulations. In the first, we evaluate the system performance with respect to number of queues ($L$) and servers ($K$) as well as channel connectivity (Figures \ref{fig:SR} to \ref{figSR3}). Arrivals  are assumed to be i.i.d. Bernoulli.
In the second group  (Figures \ref{figSR41} to \ref{figSR43}) we consider batch arrivals with random (uniformly distributed) burst size.

The policies used in this simulation are: LCSF/LCQ,
as an approximation of an MB policy; MCSF/SCQ,
as an approximation of an LB policy. An MB policy was implemented using full search, for the  cases specified in this section, and its performance was indistinguishable from that of the LCSF/LCQ. Therefore, in the simulation graphs the MB and LCSF/LCQ are represented by the same curves. This statement is also true for  LB and MCSF/SCQ policies performances. Other policies that were simulated include the  randomized, Most Connected Server First/Longest Connected Queue (MCSF/LCQ), and Least Connected Server First/Shortest Connected Queue (LCSF/SCQ) policies. The randomized policy is the one that, at each time slot, allocates each server  randomly and with equal probability  to one of its connected queues.
The MCSF/LCQ policy differs from the LCSF/LCQ policies in the order that it allocates the servers. It uses the exact reverse order, starting the allocation with the most connected server and ending it with the least connected one. However, it resembles the LCSF/LCQ policies in that it allocates each server to its longest connected queue. The  LCSF/SCQ policy allocates each server, starting from the one with the least number of connected queues, to its shortest connected queue. The difference from an LCSF/LCQ policy is obviously the allocation to the shortest connected queue. This policy will result in greatly unbalanced queues and hence a performance that is closer to the LB policies.

Figure \ref{fig:SR} shows the average total queue occupancy versus arrival rate under the five different policies. The system in this simulation is a symmetrical system with 16 parallel queues ($L=16$), 16 identical servers ($K=16$) and i.i.d. Bernoulli queue-to-server (channel) connectivity with parameter $p=P[G_{i,j}(t)=1]=0.2$.

\begin{figure}[t]
\centering
\includegraphics [width=2.8in] {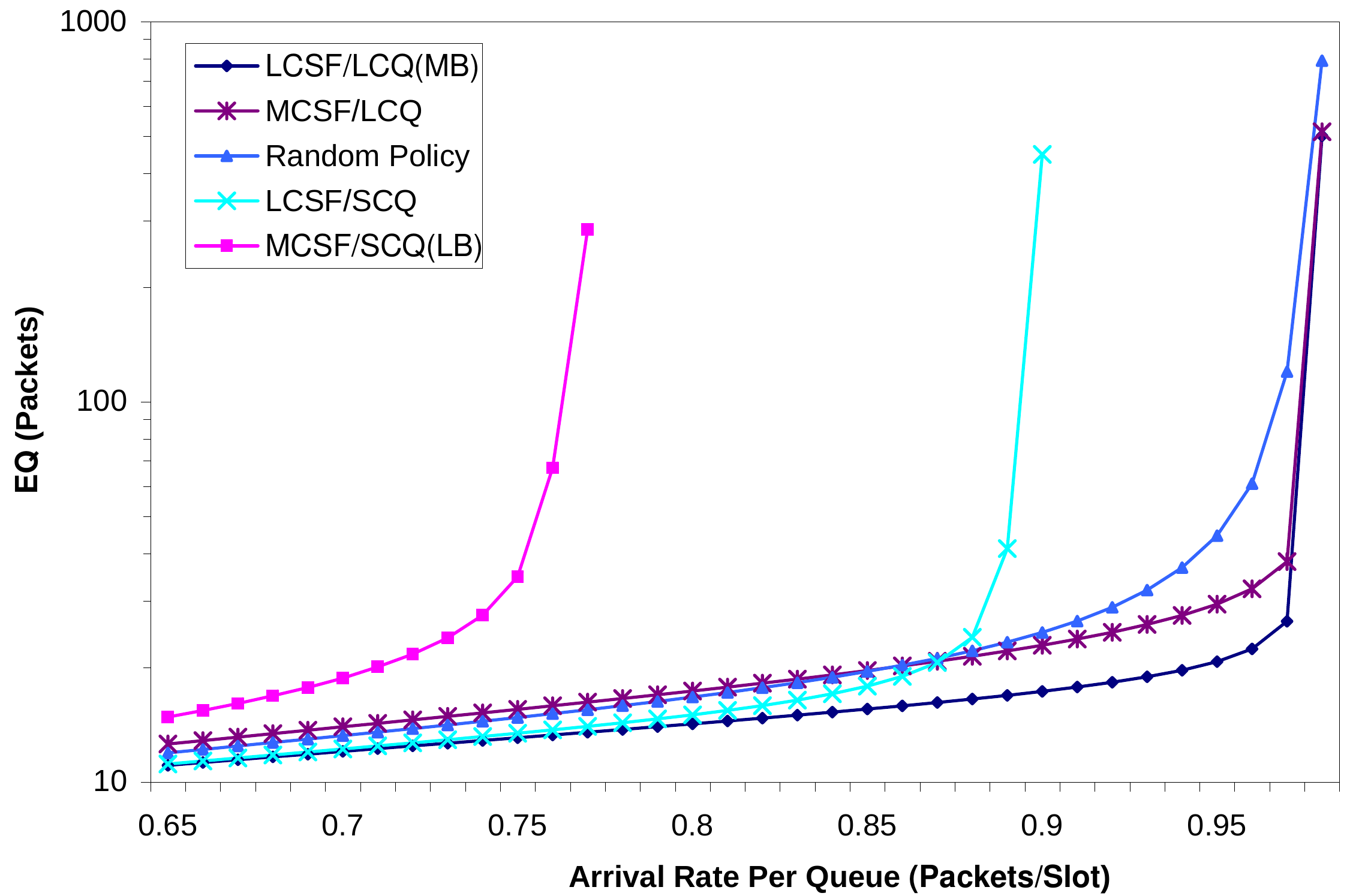}
\caption{Average total queue occupancy, $EQ$, versus load under different policies, $L=16, K=16$ and $p=0.2$.}
\label{fig:SR}
\end{figure}

The curves in Figure \ref{fig:SR} follow a  shape that is initially almost flat  and ends with a rapid increase. This abrupt increase happens at a point where the system becomes unstable. In this case, the queue lengths in the system will grow fast and the system becomes unstable. The graph shows that LCSF/LCQ, the MB policy approximation outperforms\footnote{99\% confidence intervals are very narrow and would affect the readability of the graphs. Therefore they are not included.}
all other policies. It minimizes $EQ$ and hence the queuing delay. We also noticed that it maximizes the system stability region and hence the system throughput as well.
The MCSF/SCQ performed the worst.
As expected, the performance of the other three policies lies within the performance of the MB and LB policies.


The MCSF/LCQ and LCSF/SCQ policies are variations of the MB and LB policies respectively. The performance of MCSF/LCQ policy is close to that of the MB policy. The difference in performance is due to the order of server allocation. On the other hand, the LCSF/SCQ  policy shows a large performance improvement on that of the LB policy. This improvement is a result of the reordering of server allocations.

Figure \ref{fig:SR} also shows that the randomized policy performs reasonably well. Moreover, its performance improves as the number of servers in the system decreases, as the next set of experiments shows.

\subsection{The Effect of The Number of Servers}
In this section, we study the effect of the number of servers on policy performance. Figure \ref{fig:SR11} ($K=8$) and Figure \ref{fig:SR12} ($K=4$) show $EQ$ versus arrival rate per queue under the five policies, in a symmetrical system with $L=16$ and $p=0.2$.
Comparing these two graphs to the one in Figure \ref{fig:SR}, we notice the following:

First, the performance advantage of the LCSF/LCQ (and hence of an MB policy) over the other policies increases as the number of servers in the system increases. The presence of more servers implies that the server allocation action space is larger. Selecting the optimal (i.e., MB) allocation, over any arbitrary policy, out of a large number of options will  result in a better performance as compared to the case when the number of server allocation options is less.

Second, the stability region of the system becomes narrower when less servers are used. This is true because fewer resources (servers) are available to be allocated by the working policy in this case.

Finally, we notice that the MCSF/LCQ performs very close to the LCSF/LCQ  policy in the case of $K=4$. Apparently, when $K$ is small, the order of server allocation does not have a big impact on the policy performance.

\begin{figure}[t]
\centering
\includegraphics [width=2.8in] {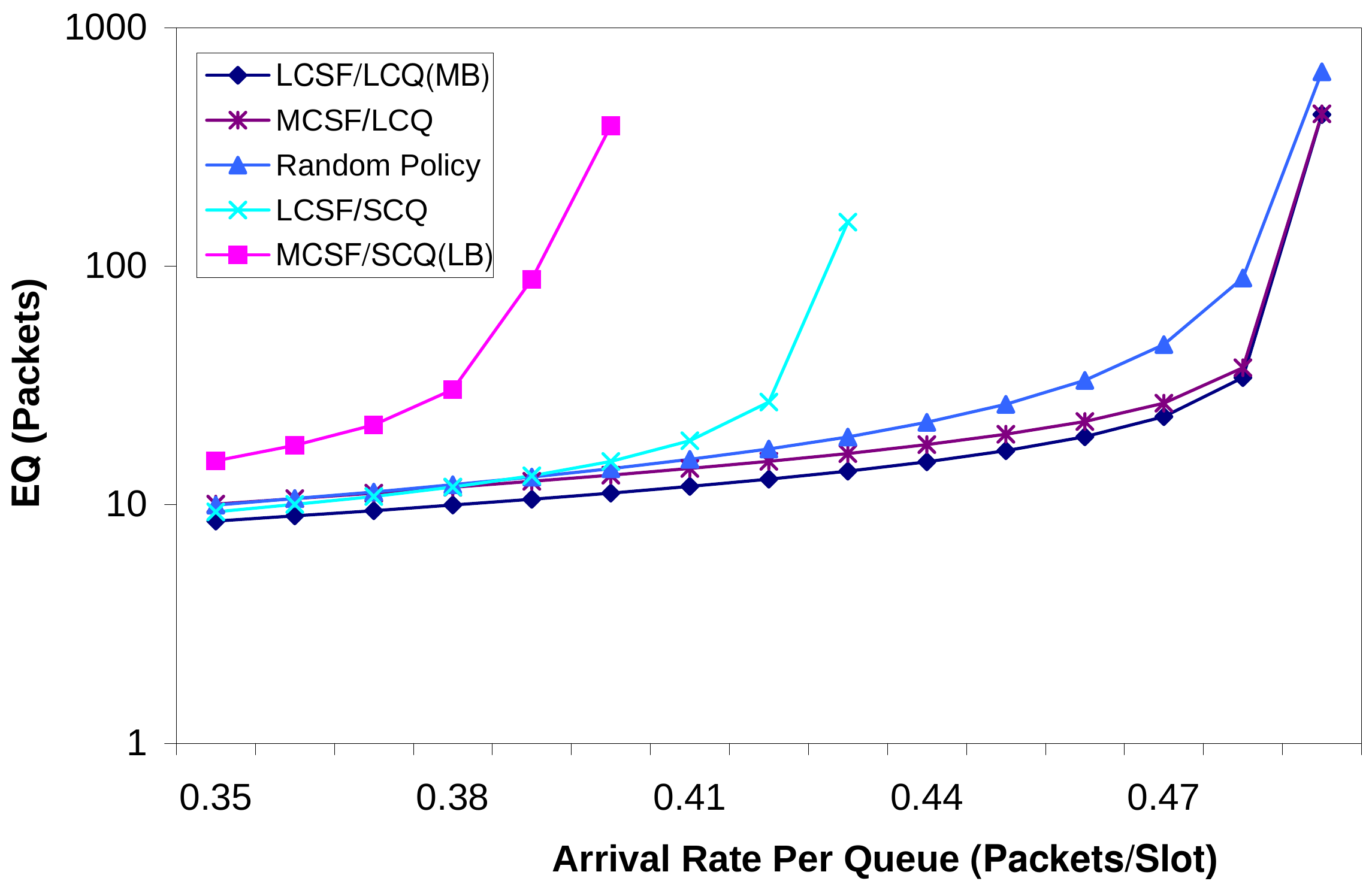}
\caption{Average total queue occupancy, $EQ$,  versus load, $L=16, K=8$ and $p=0.2$.}
\label{fig:SR11}
\end{figure}

\begin{figure}[t]
\centering
\includegraphics [width=2.8in] {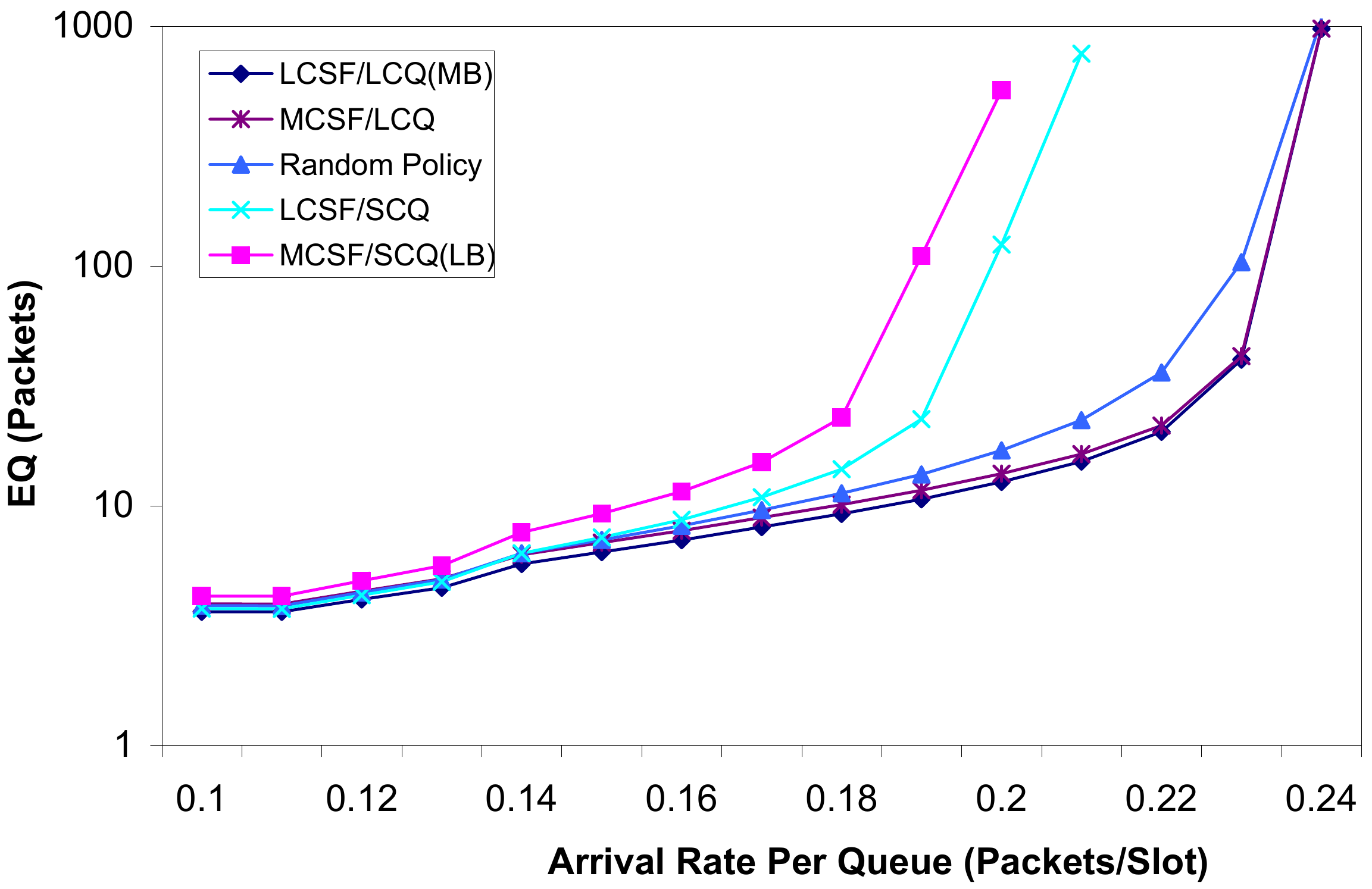}
\caption{Average total queue occupancy, $EQ$,  versus load, $L=16, K=4$ and $p=0.2$.}
\label{fig:SR12}
\end{figure}

\subsection{The Effect of Channel Connectivity}
In this section we investigate the effect of channel connectivity on the performance of the previously considered policies. Figures \ref{figSR2} and \ref{figSR3} show this effect for two choices of $L$ and $K$. We observe the following:

First, we notice that for larger channel connection probabilities ($p\geq 0.9$), the effect of the policy behavior on the system performance becomes less significant. Therefore, the performance difference among the various policies is getting smaller. The LCSF/LCQ  policy still has a  small advantage over the rest of the policies, even though it is statistically difficult to distinguish. MCSF/SCQ continues to have the worst performance. As $p$ increases, the probability that a server will end up connected  to a group of empty queues will be very small regardless of the policy in effect. In fact, when the servers have full connectivity to all queues (i.e., $p=1.0$) we expect that any work conserving policy will minimize the total number of packets in a symmetrical homogeneous system of queues since, any (work-conserving) policy will be optimal in a system with full connectivity.

Second, from all graphs we observe that there is a maximum input load that results in a stable system operation (maximum stable throughput)\footnote{The last point in every curve corresponds to an overloaded system operating beyond its stability region. As a result the simulation is permanently in a ``transient'' state. Such points are shown in the presented graphs for illustrative purposes in order to show the trend of the queue size.}. An upper bound (for stable system operation) for the arrival rate per queue $\alpha$ is given by
\begin{equation} \label{eq:stability}
    \alpha < \frac{K}{L} \big(1-(1-p)^L \big),
\end{equation}
i.e., the average number of packets entering the system ($\alpha  L$) must be less than the rate they are being served.
When $p=1.0$, the stability condition in Inequality (\ref{eq:stability}) will be reduced to $\alpha L < K$, which makes intuitive sense in such a system.

Finally, we observe that the MCSF/LCQ policy performance is very close to that of LCSF/LCQ. However, its performance deteriorates in systems with higher number of servers and lower probabilities for queue-server connectivity. It is intuitive that with more servers available, the  effect of the order of server allocations  on the policy performance will increase. Since MCSF/LCQ differs from LCSF/LCQ  only by the order of server allocation, therefore, more servers implies larger performance difference. Also, the lower the connectivity probability, the higher the probability that a server will end up with no connectivity to any non-empty queue, and hence be forced to idle.

\begin{figure*}[t]
  \renewcommand{\Diamond}{\times}
  \begin{center}
    \subfigure[$p=0.3$]{\label{figSR21}\includegraphics[width=5.8cm] {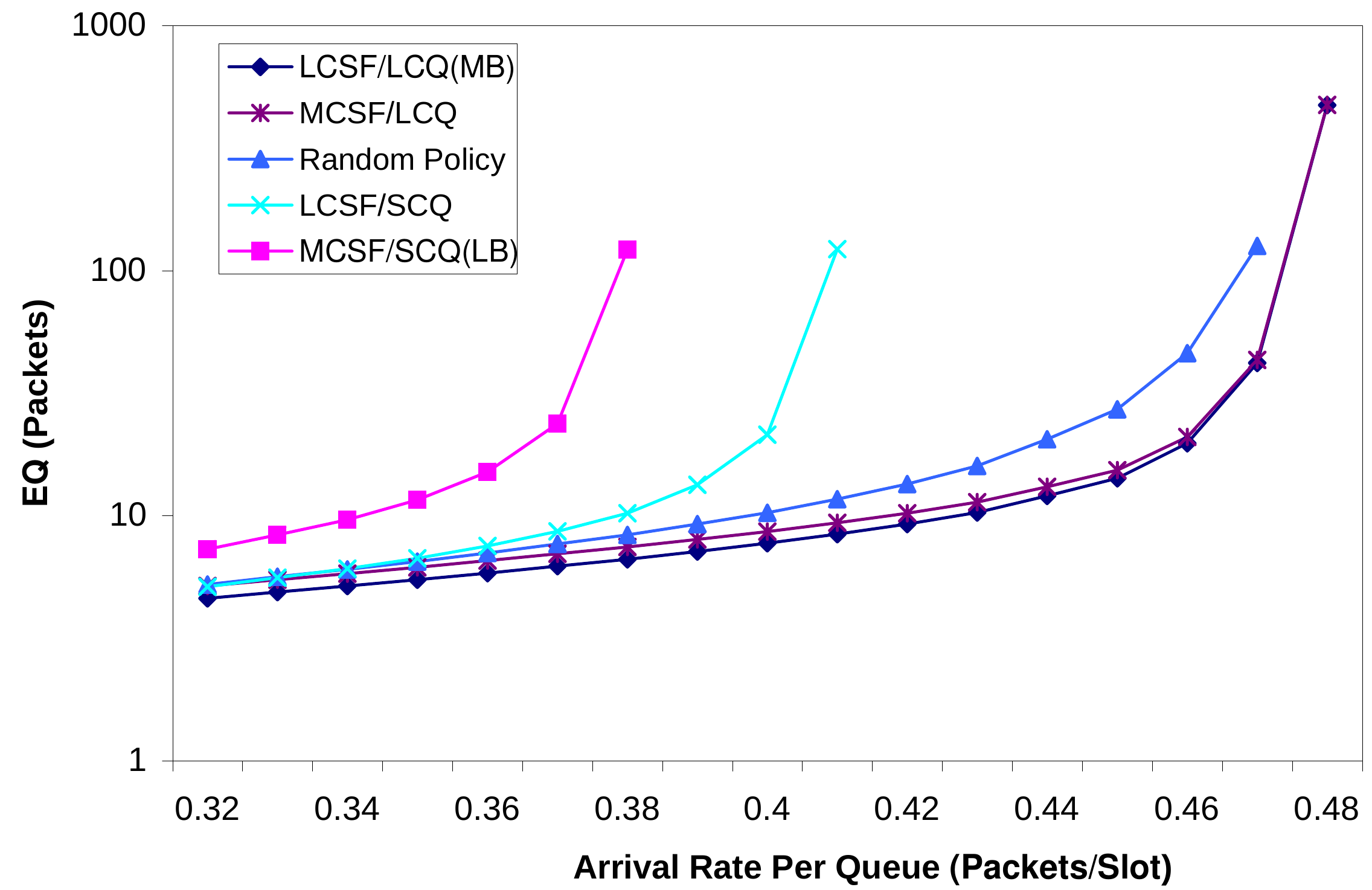}} \hspace{-.1cm}
    \subfigure[$p=0.5$]{\label{figSR22} \includegraphics[width=5.8cm]{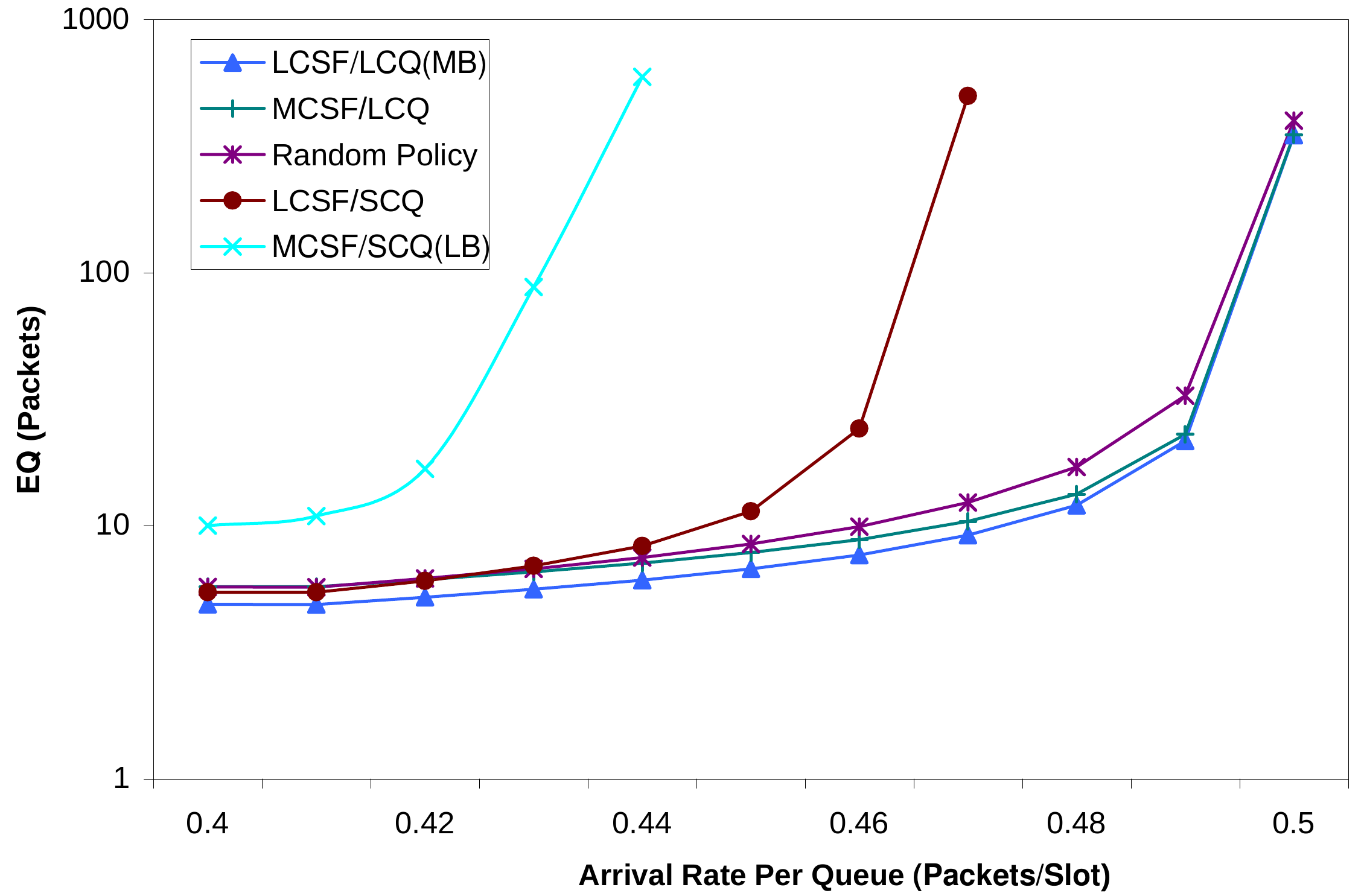}} \hspace{-.1cm} 
    \subfigure[$p=0.9$]{\label{figSR23} \includegraphics[width=5.8cm]{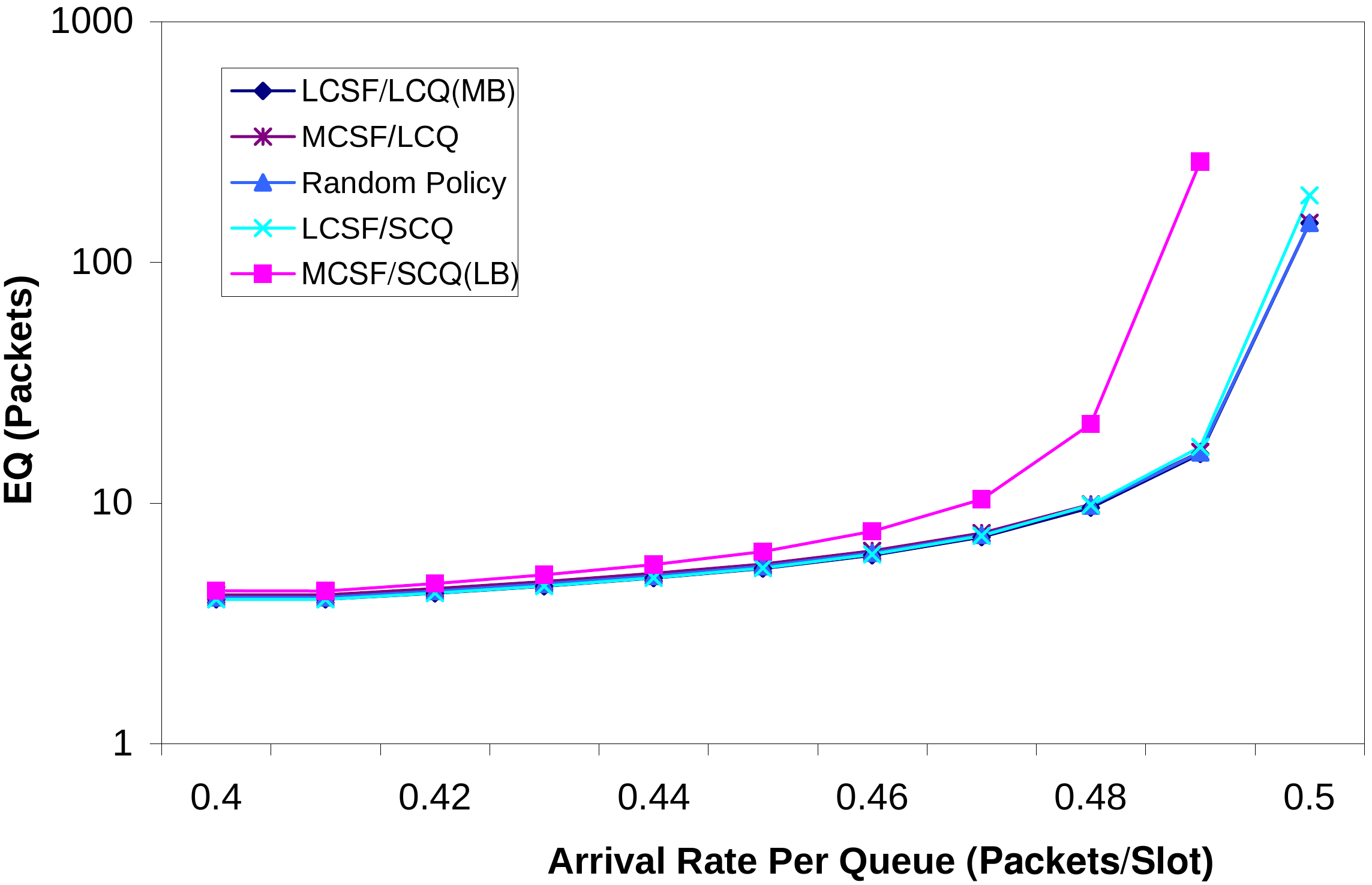}}\\
    \end{center}
  \caption{Average total queue occupancy, $EQ$,  versus load under different policies, $L=8$ and $K=4$.}
\label{figSR2}
  \vspace{-1 mm}
\end{figure*}

\begin{figure*}[t]
  \renewcommand{\Diamond}{\times}
  \begin{center}
    \subfigure[$p=0.3$]{\label{figSR31}\includegraphics[width=5.8cm] {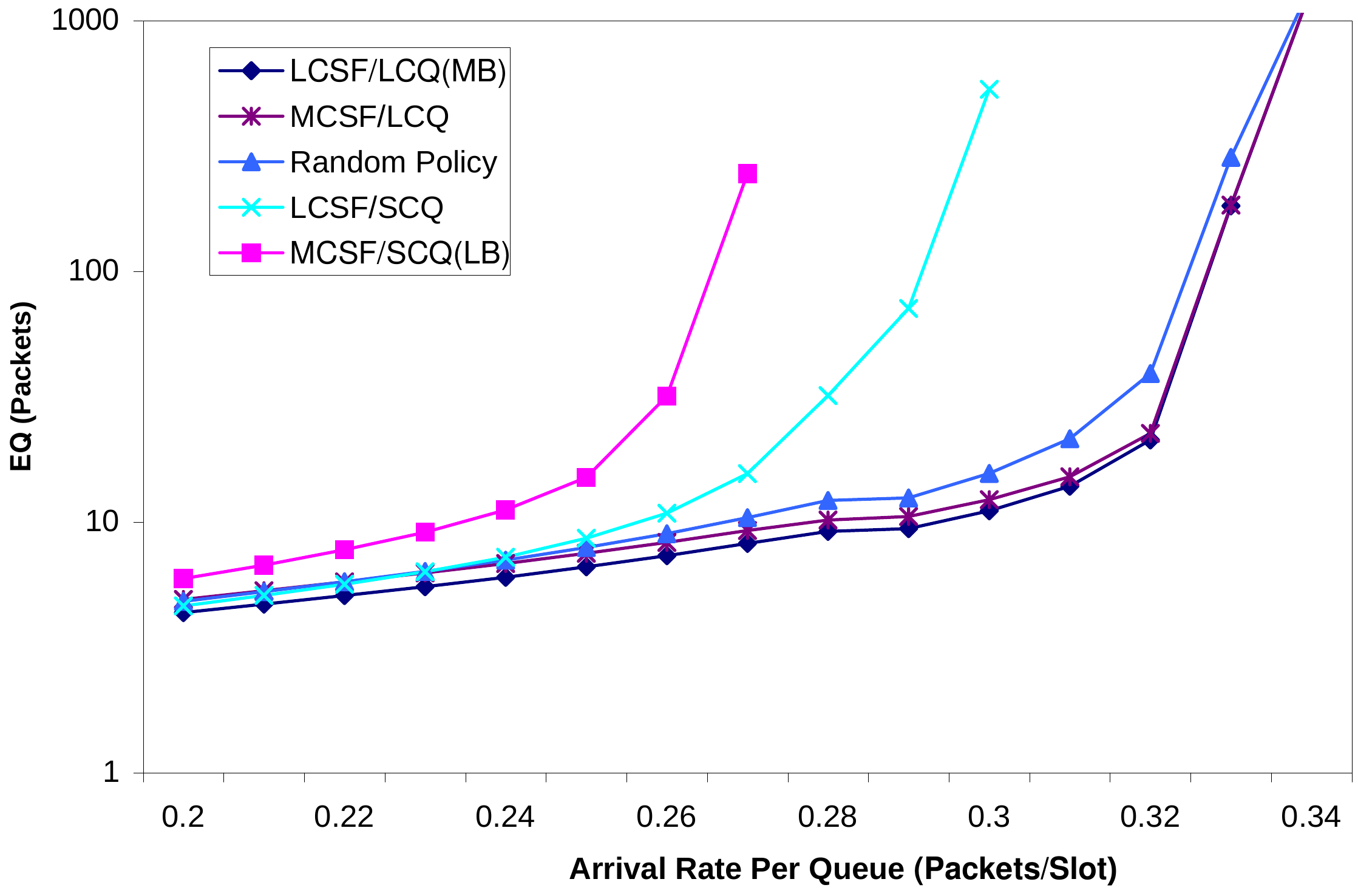}} \hspace{-.1cm}
    \subfigure[$p=0.5$]{\label{figSR32} \includegraphics[width=5.8cm]{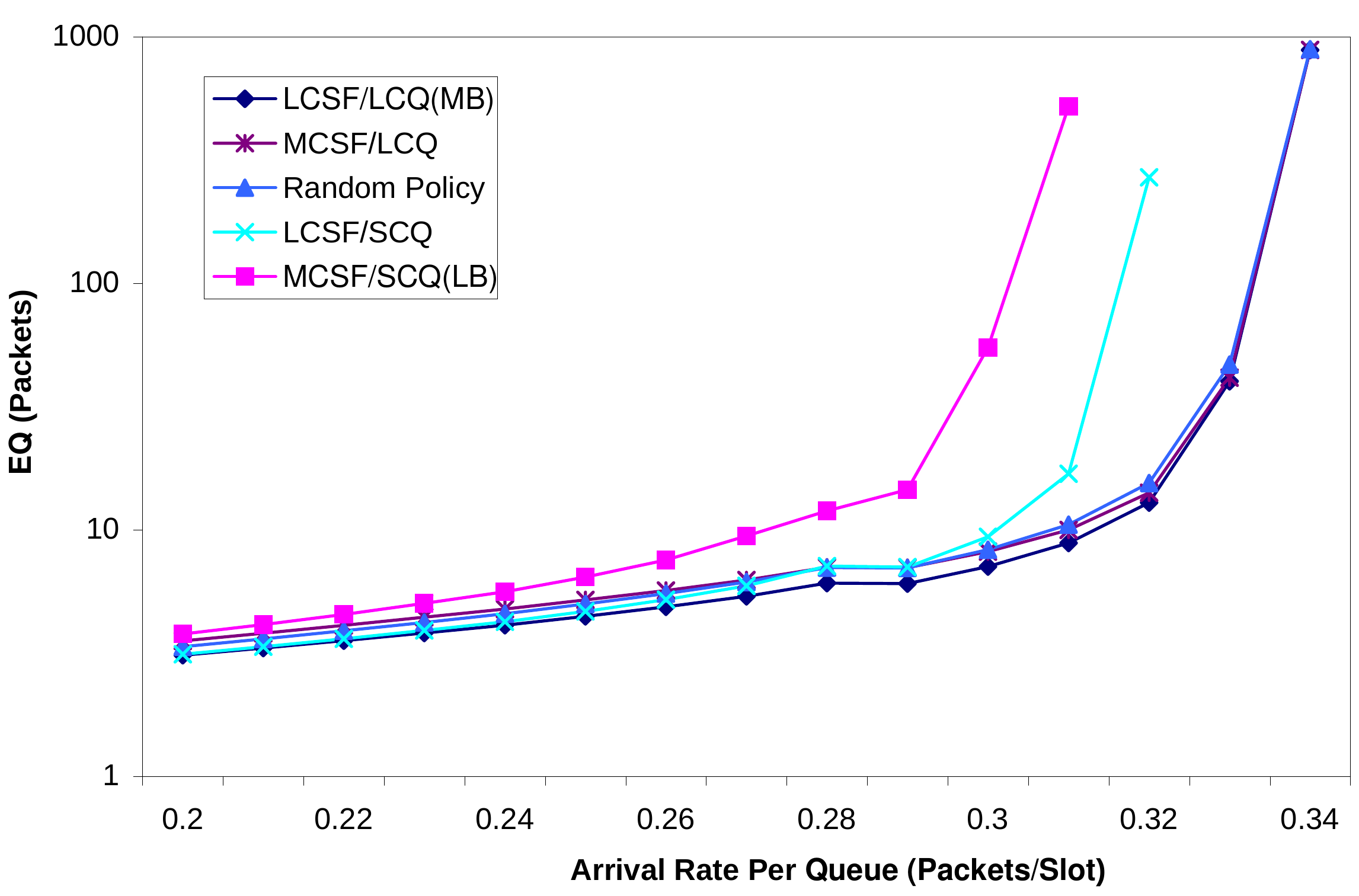}} \hspace{-.1cm} 
    \subfigure[$p=0.9$]{\label{figSR33} \includegraphics[width=5.8cm]{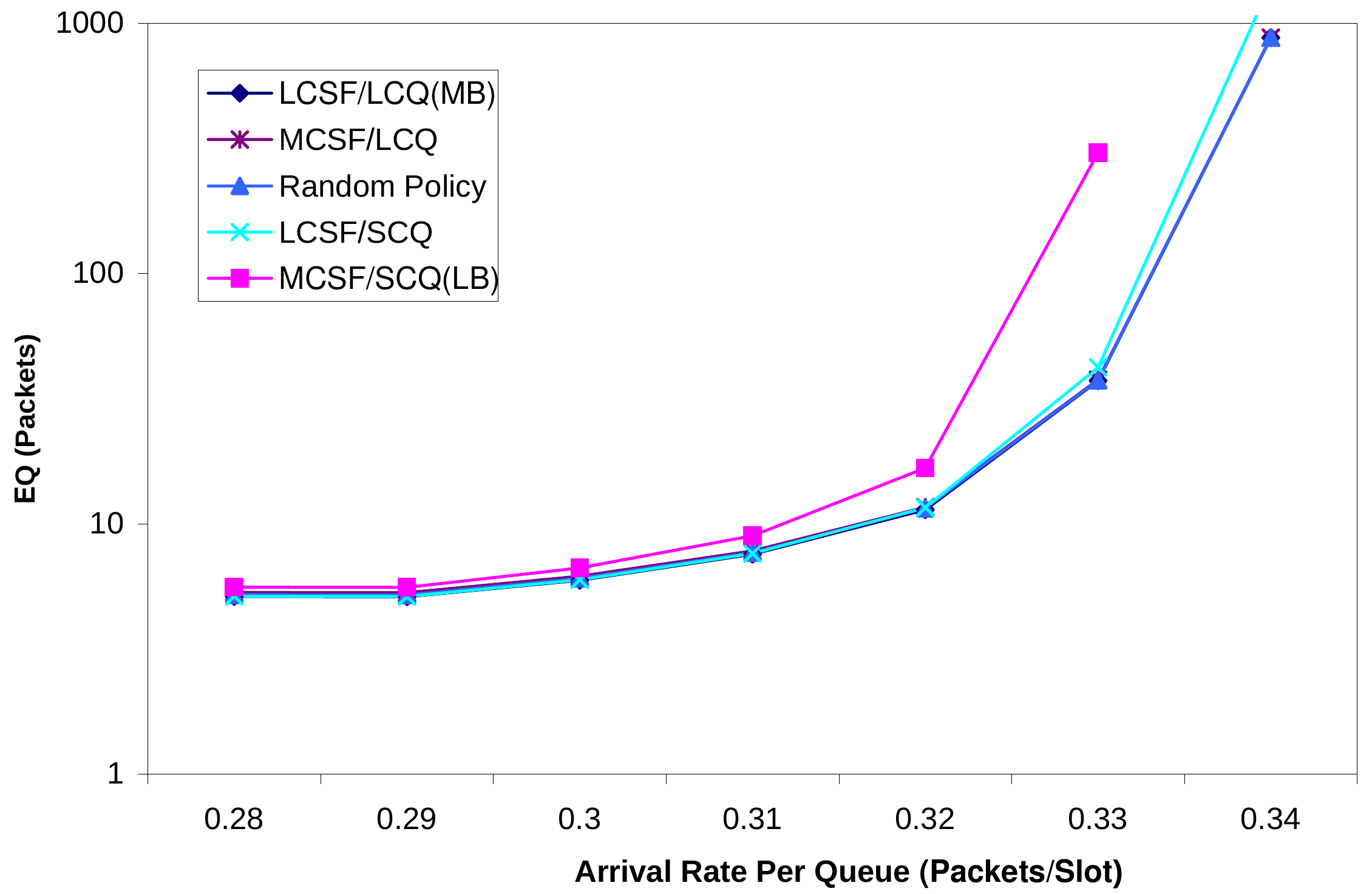}}\\
 \end{center}
 \caption{Average total queue occupancy, $EQ$,  versus load under different policies, $L=12$ and $K=4$.}
 \label{figSR3}
\end{figure*}

\subsection{Batch Arrivals With Random Batch Sizes}
We studied the performance of the presented policies in the case of batch arrivals with uniformly distributed batch size, in the range $\{1,\dots,U\}$. Figure \ref{figSR4} shows $EQ$ versus load for three cases with $U=2,5,10$, and hence average batch sizes 1.5, 3, and 5.5. The LCSF/LCQ policy clearly dominates all the other policies. However, the performance of the other policies, including MCSF/SCQ (LB approximation) approaches that of the LCSF/LCQ policy as the average batch size increases. The performance of all the policies deteriorates when the arrivals become burstier, i.e., the batch size increases. 

\begin{figure*}[th]
  \renewcommand{\Diamond}{\times}
  \begin{center}
    \subfigure[$p=0.5$, average batch size = 1.5.]{\label{figSR41}\includegraphics[width=5.8cm] {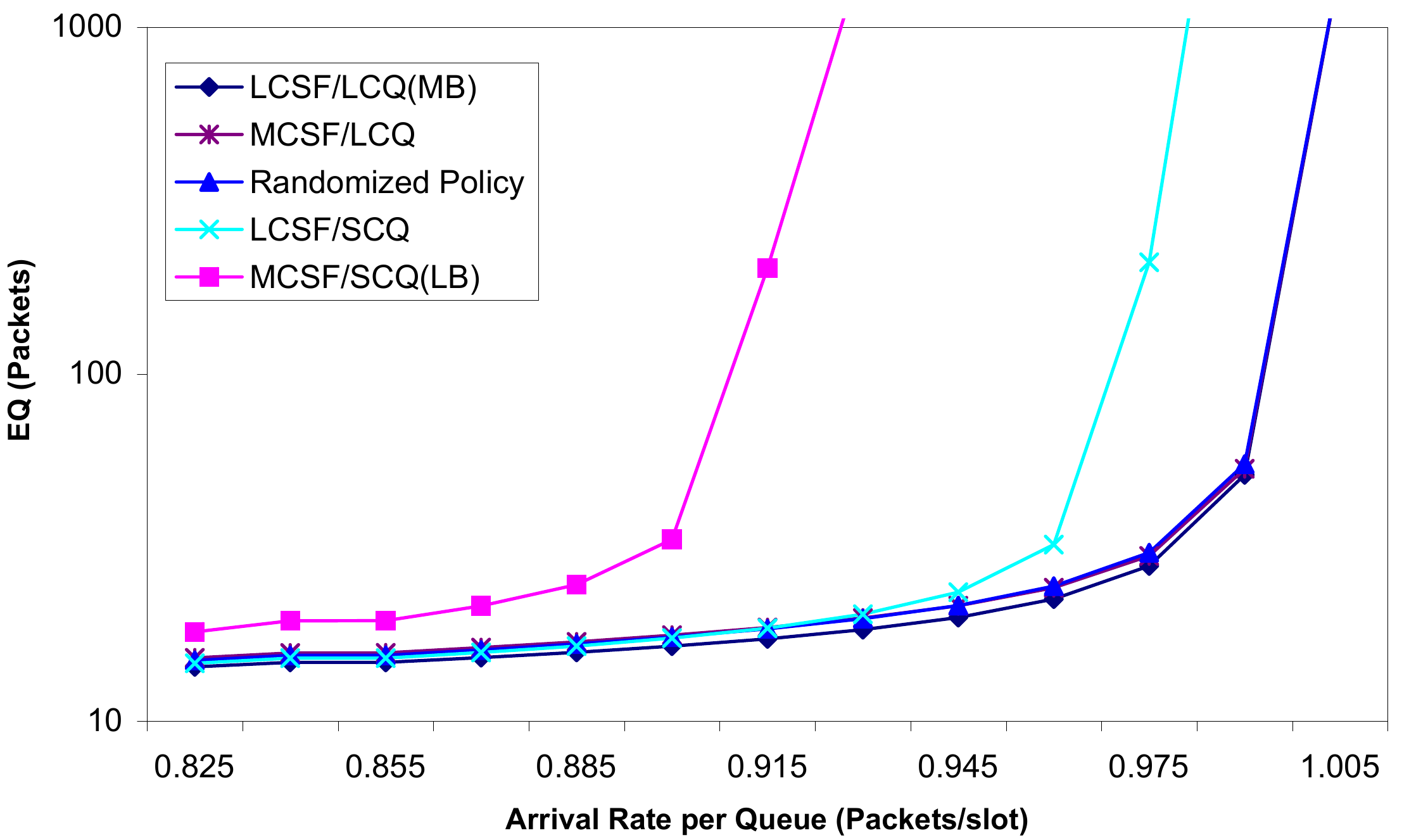}} \hspace{-.1cm}
    \subfigure[$p=0.6$, average batch size = 3.]{\label{figSR42} \includegraphics[width=5.8cm]{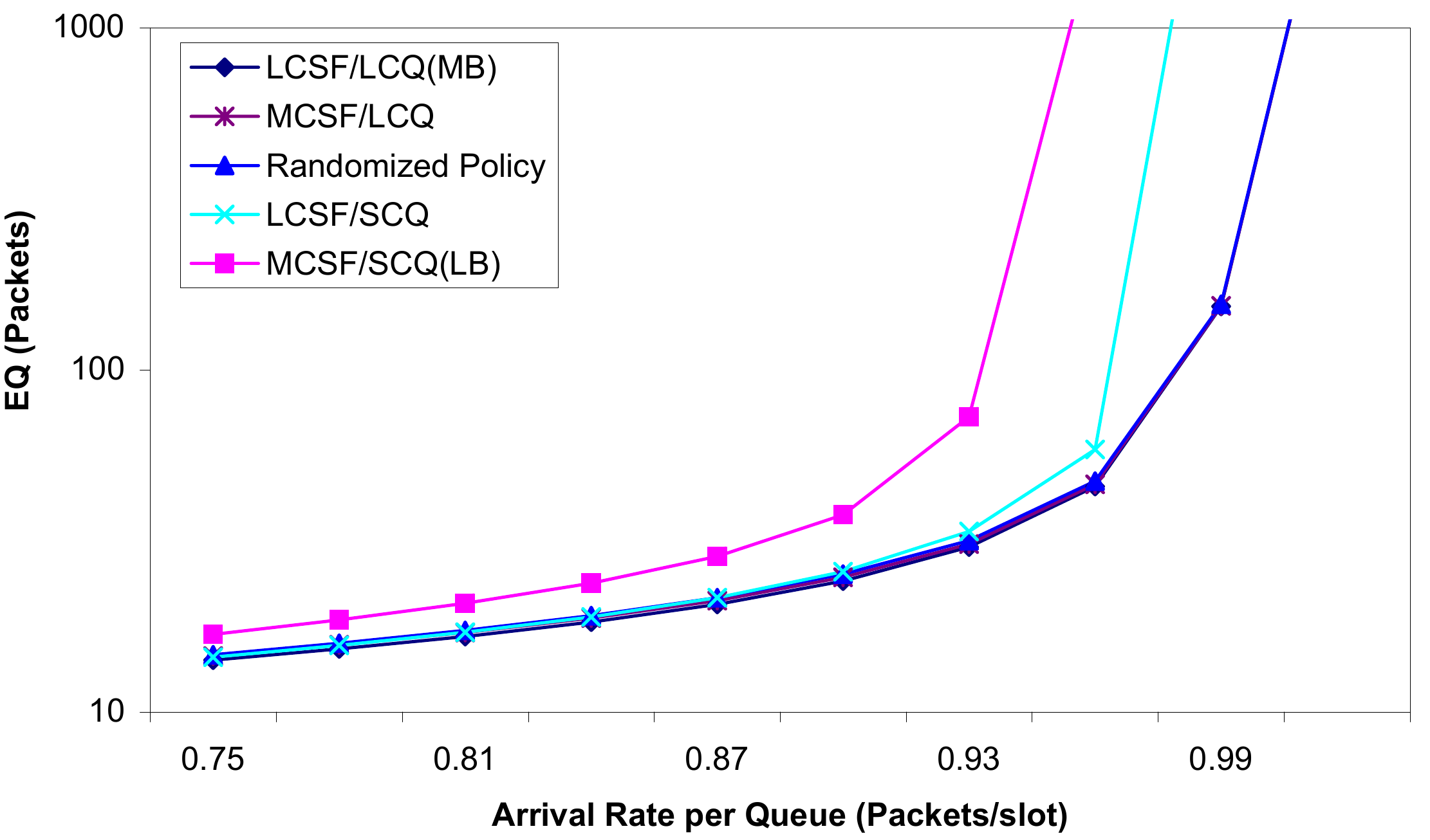}} \hspace{-.1cm} 
    \subfigure[$p=0.8$, average batch size = 5.5.]{\label{figSR43} \includegraphics[width=5.8cm]{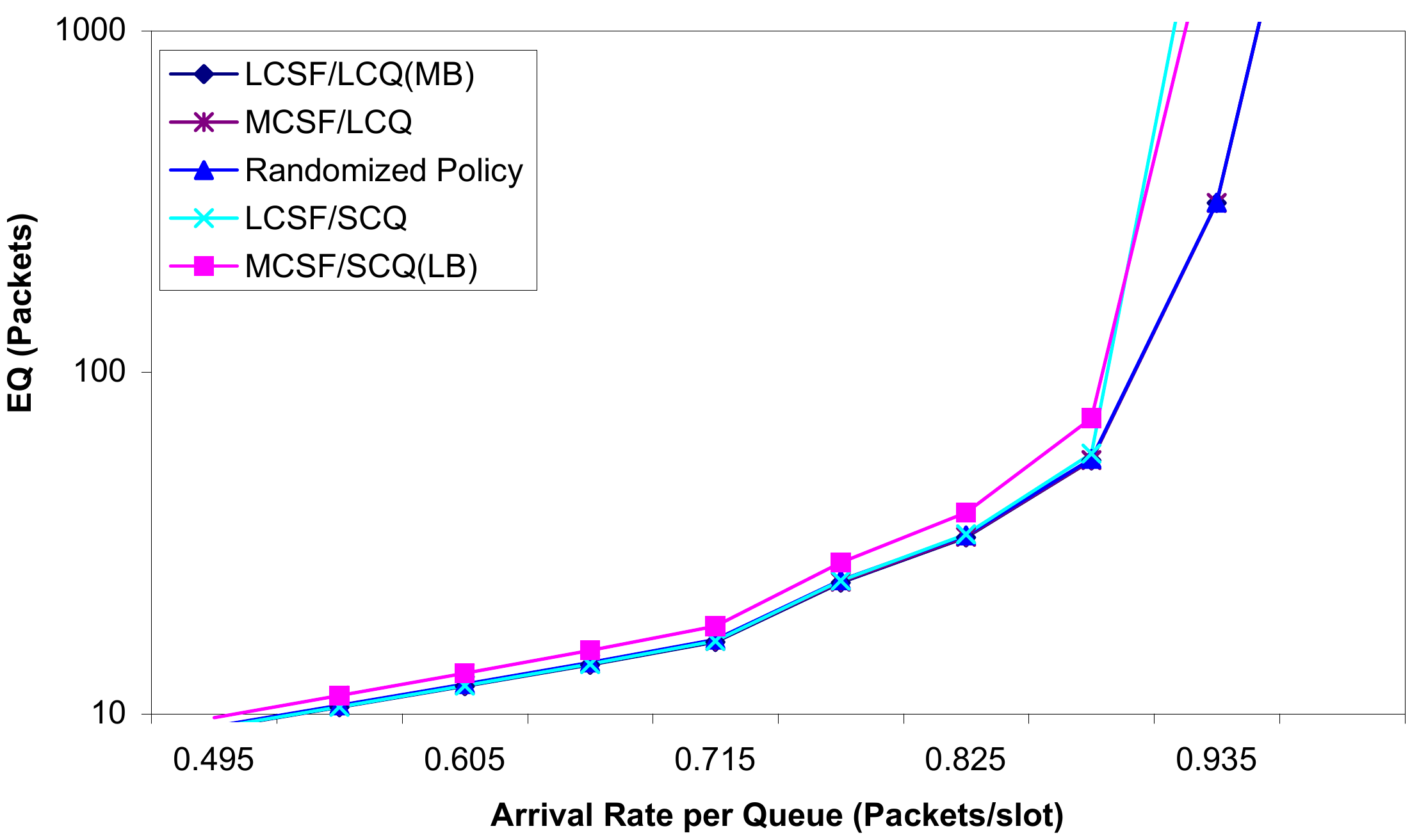}}\\
   \end{center}
  \caption{Average total queue occupancy, $EQ$,  versus load, batch arrivals, $L=16$ and $K=16$.}
\label{figSR4}
\end{figure*}


\section{Final Remarks}\label{sec:final-remarks}
The model and the results presented in this article can be regarded as a generalization (with the obvious added complexity as well as  utility of our model) of the models and results reported by \cite{Tassiulas}, \cite{ganti}, and  \cite{Kittipiyakul}.  

In  \cite{Tassiulas}, the authors investigated the optimal scheduling policy for a model of $L$ parallel queues  and one randomly connected server. This model is a special case of the model we presented in this article, i.e., when $K=1$. Using stochastic dominance techniques, they proved that LCQ is optimal in that it minimizes the total number of packets in the system.  In our work, we also use stochastic dominance  techniques to prove the optimality of MB policies for a wide range of cost functions (cost functions that are monotone, non-decreasing with respect to the partial order $\prec_p$) including the total number of packets in the system. It can be easily shown that for the case of a single server (i.e., $K=1$) the LCQ policy minimizes the imbalance index and therefore, LCQ belongs to the set of MB policies.

In   \cite{ganti}, the authors investigated the optimal policy for a model of $L$ parallel queues with a stack of $K$ servers. Each queue is  randomly connected to the  entire server stack. Only one server can be allocated to a queue at any time slot. In contrast, our model assumes independent queue-server connectivity, i.e., a queue can be connected to a subset of the $K$ servers and not connected to the rest at any given time slot. We also allow for multiple  servers to be allocated (when connected) to any queue. Therefore, the model in  \cite{ganti} can also be considered as a special case of our model, i.e., by letting $g_{i,j}(t) = g_i(t), \forall i,j,t$ and by adding the feasibility constraint $y_i(t) \leq 1$. They proved that a policy  that allocates the $K$ servers to the $K$ longest connected queues (LCQ) is optimal. Under the constraints above, this policy would also minimize the imbalance index among all feasible policies, i.e., this policy belongs to the set of MB policies.

In  \cite{Kittipiyakul} the authors proved that, in a model  of two parallel queues ($L=2$) and multiple randomly connected servers, a MTLB (maximum throughput/load balancing) policy minimizes the expected total cost. They defined the cost as a class of functions of the queue lengths for the two queues in the system.  In our work, we generalize the model in \cite{Kittipiyakul} as follows: (a) we extend the model to   $L>2$, (b) we optimize the cost function in the  stochastic order sense which implies the expected total cost used in  \cite{Kittipiyakul}, and (c) we relax the   supermodularity and  convexity constraints that they enforced on the cost function, i.e., we prove our results for a larger set of cost functions that includes theirs.

The authors of  \cite{Kittipiyakul} defined the MTLB policy as the one that minimizes the lexicographic order of the queue length vector while maximizing the instantaneous throughput. We can show that MTLB policy belongs to the set of MB policies.
To  do that, we have to show that a policy which minimizes the lexicographic order: (a) also minimizes the imbalance index, i.e., it belong to the set of MB policies, and (b) is a work-conserving policy. A work-conserving policy minimizes the number  of idling servers and hence, it maximizes instantaneous throughput (by the definition of instantaneous throughput).
Lemma \ref{lem:lexIsMB} states these results formally. 

\begin{lem}\label{lem:lexIsMB} 
Given the state $(\mathbf x(n), \mathbf g(n)) $ during time slot $n$. Let $\mathbf{\lambda}^*$ be a vector resulted from the feasible withdrawal vector $\mathbf y^*(n)\in \mathcal Y(\mathbf x(n), \mathbf g(n)) $. 
Suppose that $\mathbf{\lambda}^* \leq_{lex} \mathbf{\lambda}$ for all feasible $\mathbf{\lambda}$. Then: a) the  vector $\mathbf {\lambda}^*$ achieves  the minimum imbalance index among all feasible  vectors, and b) A policy that selects $\mathbf y^*(n)$ is a work-conserving policy.
\end{lem}

\begin{proof}
a) 
Assume to the contrary that $\mathbf {\lambda}^*$ does not minimize the imbalance index. Then there must exist a $\mathbf y'(n)\in \mathcal Y(\mathbf x(n), \mathbf g(n)), \mathbf y'(n) \neq \mathbf y^*(n) $ such that the imbalance index of the vector $\mathbf {\lambda}'$ is strictly less than that of $\mathbf {\lambda}^*$. This implies that a policy $\pi^*$ that results in the withdrawal vector $\mathbf y^*(n) $, and therefore the vector  $\mathbf {\lambda}^*$, belongs to the set $\Pi_n^h \setminus \Pi_n$ for some $h>0$, i.e., it does not have  the MB property during time slot $n$. For any given state, a policy that minimizes the imbalance index must exist, `minimization on a finite set'. According to Lemma 4,   $h_{\pi^*}$ balancing interchanges (BIs are feasible interchanges) are required to make any policy in $\Pi_n^h$ belongs to $ \Pi_n$. Lemma D-1 shows that such balancing interchanges  are feasible. 
Therefore, the following \textit{balancing interchange}  is both feasible and enhancing (it reduces the imbalance index):
\begin{equation}\label{eq:1}
		\mathbf {\lambda}' =\mathbf {\lambda}^* - \mathbf I(l,s),
\end{equation}
for any $l,s \in \{0,1,\ldots, L\}$ and $\lambda^*_s < \lambda^*_l -1$. 

In other words,  we perform a feasible server reallocation from the $s^{th}$ longest queue to the $l^{th}$ longest queue in the system during time slot $n$. The resulted leftover vector  $\mathbf {\lambda}'$ is related to the vector $\mathbf {\lambda}^*$ as follows:
\begin{equation}
    \lambda_l' = \lambda^*_l -1, \quad \lambda_s' = \lambda^*_s +1, \quad \lambda_i' = \lambda^*_i, \quad \forall i \neq l,s.
\end{equation}

Since $l<s$ by definition, then it is clear that $\mathbf{\lambda}' \leq_{lex} \mathbf{\lambda^*}$. This contradicts the initial assumption. Therefore, $\mathbf {\lambda}^*$ must have  the minimum imbalance index.

b) A feasible interchange $\mathbf y'(n) = \mathbf y(n) + \mathbf I(f,0)$ is a balancing one, since by definition of queue 0 and the interchange feasibility  conditions, we have $\hat x_f(n)> \hat x_0(n)+1$. Queue 0 is permanently connected to all servers by assumption. According to Lemma B-1 this interchange will definitely reduce the imbalance index. Therefore, any policy that intentionally idles servers can always be improved (i.e., reduce its  imbalance index) by using the balancing interchange $\mathbf I(f,0)$ for some queue $f \in \{1,\ldots, L \}$.

From part a) of this lemma, we showed that a policy that minimizes the lexicographic order will also minimize the imbalance index. We also showed that a  policy that idles servers intentionally  can not achieve the minimum imbalance index. Therefore, only a work-conserving policy can minimize the lexicographic order.
\end{proof}



From the above, we conclude that the MTLB belongs to the class of MB policies.


\section{Conclusion}
In this work, we presented a model for dynamic packet scheduling in a multi-server systems with random connectivity. This model can be used to study packet scheduling  in emerging wireless systems. We modeled such systems via symmetric queues with random server connectivities and and Bernoulli arrivals. We introduced the class of Most Balancing policies. These policies distribute the service capacity between the  connected queues in the system in an effort to ``equalize'' the queue occupancies. A theoretical proof of the optimality of MB policies using stochastic coupling arguments was presented. Optimality was defined as minimization, in stochastic ordering sense, of a range of cost functions of the queue lengths. The LCSF/LCQ policy was proposed as good, low-complexity approximation for MB policies.

A simulation study was conducted to study the performance of five different policies. The results verified that the MB approximation outperformed all other policies (even when the arrivals became bursty). However, the performance of all policies deteriorate as the mean burst size increases. Furthermore, we observed (through simulation) that the performance gain of the optimal policy over the other policies is reduced greatly in this case.
Finally, we observed that a randomized policy can perform very close to the optimal one in several cases.


\appendix
\renewcommand{\thesection}{Appendix \Alph{section}}
\renewcommand{\theequation}{A-\arabic{equation}}
 \setcounter{equation}{0}  
 \renewcommand{\thefigure}{A-\arabic{figure}}
 \setcounter{figure}{0}  

\subsection{Proof of Lemma  \ref{lem:I1}}\label{appendixB1iii}
\begin{proof} 
To prove part (a), assume that $\mathbf D=0$; then, using Equation (\ref{eq:I1}), we have:
\begin{eqnarray}
	  \mathbf y(n) &=& \mathbf y^{MB}(n) \nonumber \\
	  \mathbf x(n) - \mathbf y(n) &=& \mathbf x(n) - \mathbf y^{MB}(n) \nonumber \\
	  \mathbf {\hat x}(n) &=& \mathbf {\hat x}^{MB}(n) \label{eq:I11}
\end{eqnarray}
From Equations (\ref{eq:I11}) and (\ref{eq:kappa}), we have that  $\kappa_n(\pi) = \kappa_n(\pi^{MB})$
 and thus $\pi$ has the MB property during time slot $n$.

To prove  part (b), assume that $\pi$ has the MB property at time slot $n$. Therefore, $\kappa_n(\pi) = \kappa_n(\pi^{MB})$. From Lemma \ref{lem:alg1c} this is only possible if either: (i)  $\mathbf {\hat x}(n) = \mathbf {\hat x}^{MB}(n)$, or (ii) $\mathbf {\hat x}(n)$ is obtained  by performing a  balancing interchange between the pair of the $l^{th}$ and the $s^{th}$ longest queues ($l<s$) in $ \mathbf {\hat x}^{MB}(n)$ such that $\hat x_{[l]}(n) = \hat x_{[s]}(n) +1$, is satisfied; note that there may be multiple such queue pairs. The  balancing interchange in case (ii) will affect the length of two queues only (call them $i$ and $j$) such that $\hat x_i(n) = \hat x^{MB}_i(n) -1$ and $\hat x_j(n) = \hat x^{MB}_j(n) +1$, where $i=[l]$ and $j=[s]$ (for each given pair). Therefore,
\begin{eqnarray}
 y_i(n)&=& x_i(n) - \hat x_i(n)  = x_i(n) - (\hat x^{M\!B}_i(n) -1)  \nonumber \\
 	&= &y^{M\!B}_i(n) +1, \label{eq:I12}
\end{eqnarray}
and,
\begin{eqnarray}
 y_j(n) & =& x_j(n) - \hat x_j(n)  = x_j(n) - (\hat x^{M\!B}_j(n) +1)  \nonumber \\
 	    & =& y^{M\!B}_j(n) -1,  \label{eq:I13}
\end{eqnarray}
while withdrawals from the remaining queues will be the same, i.e.,
\begin{eqnarray}
 y_b(n) & =& y^{MB}_b(n), \forall b \neq i,j.  \label{eq:I13.1}
\end{eqnarray}

From Equations (\ref{eq:I12}) through (\ref{eq:I13.1}),  we conclude that the vector $\mathbf D$ has components that are $0, +1, $ or $-1$ only.
\end{proof}

\subsection{Balancing Interchanges and the Imbalance Index} \label{appendixB1}

 
In this section, we present  a lemma   that quantifies the effect of performing a balancing interchange on the imbalance index $\kappa_n(\pi)$. 

\begin{lemA}\label{lem:alg1c}
Let  $\mathbf x$ and $\mathbf x^*$  be two $L+1$-dimensional ordered vectors  (in descending order); suppose that  $\mathbf x^*$  is obtained from   $\mathbf x$ by performing a balancing interchange $ \mathbf I(l,s)$ between two components, $l$ and $s$, of $\mathbf x$, where $x_l > x_s$,
such that, $ s>l; x_l> x_a, \forall a>l$ and $x_s < x_b, \forall b<s$.
Then 
\begin{eqnarray}\label{eq:lemB1_1}
\sum_{i'=1}^{L} \sum_{j'=i'+1}^{L+1} ( x^*_{i'} - x^*_{j'} ) = \sum_{i=1}^{L} \sum_{j=i+1}^{L+1} ( x_{i} -x_{j} ) \nonumber \\
 -2(s-l) \cdot \mathds{1}_{\{ x_{l} \geq x_{s} +2  \}}
\end{eqnarray}
%
\end{lemA}

\begin{proof}

We generate the vector $\mathbf x^*$  by performing
a \textit{balancing interchange} of two components, $l$ and $s$ (i.e., the $l^{th}$ and the $s^{th}$ largest components), in the vector $\mathbf x$ and reorder the resulted vector in descending manner. The resulted vector $\mathbf x^*$ is characterized by the following:
\begin{eqnarray}\label{eq:lemB1_2}
	&& x^*_{l'}= x_{l}-1, \quad x^*_{s'}= x_{s}+1, \quad  x_{l}> x_{s} \nonumber \\
	 && x^*_{k}= x_{k}, \quad  \forall k \neq l,s, l',s'
\end{eqnarray}
where $l'$ (respectively $s'$) is the new index  (i.e., the order in the new vector $\mathbf x^*$) of   component $l$ (respectively $s$) in the original vector $\mathbf x$.



From Equation (\ref{eq:lemB1_2}) we can identify $L-2$ elements that have the same magnitude in the two  vectors $\mathbf x$ and $\mathbf x^*$. Therefore, the sum of differences between these $L-2$ elements in both vectors will also be the same, i.e.,

\begin{equation}
\sum_{\substack{i'=1\\ i' \notin  \{l',s'\}} }^{L} \sum_{\substack{j'=i'+1 \\ j' \notin  \{l',s'\}} }^{L+1} ( x^*_{i'} - x^*_{j'} ) = \sum_{\substack{i=1 \\i  \notin  \{l,s\}} }^{L} \sum_{\substack{j=i+1 \\ j  \notin  \{l,s\}} }^{L+1} ( x_{i} -x_{j} )
 \label{eq:lemB1_3}
\end{equation}

We calculate the sums for the remaining terms  (i.e., when at least one of the indices $i,j$ belongs to $ \{l,s\}$ and/or $i',j'$ belongs to $ \{l',s'\}$) next.  
We first assume that  $x_{l} \geq x_{s}+2 $; in this case, we can easily show that $l' \leq s'$. Then, we have the following five, mutually exclusive, cases to consider:

\begin{enumerate}
\item  When $i'=l', i=l, j'=s'$ and $j=s$. This case occurs only once, i.e., when decomposing the double sum in Equation (\ref{eq:lemB1_1}) we can find only one term that satisfies this case. From Equation (\ref{eq:lemB1_2}) we have
\begin{equation}\label{eq:lemB1_4}
 x^*_{l'} - x^*_{s'}  = x_{l} -x_{s} -2
\end{equation}

\item When $i'=l', i=l,j'\neq s'$ and $ j \neq s$. There are $L-l$ terms that satisfy this case. Analogous to  case 1) we determined that
\begin{equation}\label{eq:lemB1_5}
 x^*_{l'} - x^*_{j'}  = x_{l} -x_{j} -1
\end{equation}

\item When $i' \neq l', i \neq l, j'=s' $ and $ j = s$. There are $s-2$ terms that satisfy this case. In this  case we can show that
\begin{equation}\label{eq:lemB1_6}
 x^*_{i'} - x^*_{s'}  = x_{i} -x_{s} -1
\end{equation}

\item When $i'\neq l',s', i \neq l,s,  j'=l'$ and $ j = l$. There are $l-1$ terms that satisfy this case. In this  case we can show that
\begin{equation}\label{eq:lemB1_7}
 x^*_{i'} - x^*_{l'}  = x_{i} -x_{l} +1
\end{equation}

\item When $i'=s', i =s, j'\neq l',s' $ and $ j \neq l,s$. There are $L-s+1$ terms that satisfy this case. In this  case we have
\begin{equation}\label{eq:lemB1_8}
 x^*_{s'} - x^*_{j'}  = x_{s} -x_{j} +1
\end{equation}
\end{enumerate}

The above cases (i.e., Equations (\ref{eq:lemB1_3})-(\ref{eq:lemB1_8})) cover all the terms in Equation (\ref{eq:lemB1_1}) when $x_{l} \geq x_{s}+2 $. Combining all these terms yields:
\begin{eqnarray}
 \sum_{i'=1}^{L} \sum_{j'=i'+1}^{L+1} ( x^*_{i'} - x^*_{j'} ) = \sum_{i=1}^{L} \sum_{j=i+1}^{L+1} ( x_{i} -x_{j} ) \hspace{1.0cm}  \nonumber \\
 -2\cdot (1) \!-\!1 \!\cdot\! (L\!-\!l) \!-\!1 \!\cdot\! (s\!-\!2) \!+\! 1\! \cdot\! (l\!-\!1)\! +\! 1 \!\cdot\! (L\!-\!s\!+\!1) \nonumber \\
 = \sum_{i=1}^{L} \sum_{j=i+1}^{L+1} ( x_{i} -x_{j} ) -2 (s-l) 
 \label{eq:lemB1_9}
\end{eqnarray}

Furthermore, if $x_{l} =x_{s}+1 $, then from Equation (\ref{eq:lemB1_2}) it is  clear that $x^*_{l'} =x_{s}$ and $x^*_{s'} =x_{l}$, i.e., the resulted vector is a permutation of the original one. Therefore,  the sum of differences will be the same in both vectors and Equation (\ref{eq:lemB1_1}) will be reduced to

\begin{equation}\label{eq:lemB1_91}
\sum_{i'=1}^{L} \sum_{j'=i'+1}^{L+1} ( x^*_{i'} - x^*_{j'} ) = \sum_{i=1}^{L} \sum_{j=i+1}^{L+1} ( x_{i} -x_{j} )
\end{equation}
Equation (\ref{eq:lemB1_1}) follows from Equations (\ref{eq:lemB1_9}) and (\ref{eq:lemB1_91}). 
\end{proof}

\subsection{Proof of Lemma  \ref{lem:I4}}\label{appendixB1i}

%
 
We first  introduce a few intermediate lemmas that describe  properties of  $\mathbf I(f,t)$ and  $\mathbf D$.

\begin{lemA} \label{lem:I003}
The feasible interchange $\mathbf I(f,0), f>0$ is a balancing interchange.
\end{lemA}

\begin{proof}
By definition, $x_0(n) =0$. Since $y_0(n) \geq 0$, therefore $\hat x_0(n) = x_0(n)- y_0(n) \leq 0$.  According to the feasibility constraint (\ref{eq:I102.1}), the interchange $\mathbf I(f,0)$ is  feasible only when $\hat x_f(n) \geq 1$. Therefore, $\hat x_f(n) \geq \hat x_0(n)+1$, and it follows that $\mathbf I(f,0)$ is a balancing interchange.
\end{proof}

\begin{lemA} \label{lem:I2}
 For a given policy $\pi \in \Pi_{n-1}$ and a time slot $n$,
\begin{equation}\label{eq:I2}
	\sum_{i=0}^L D_i \cdot \mathds 1_{\{D_i>0 \}} = - \sum_{j=0}^L D_j \cdot \mathds 1_{\{D_j < 0 \}}
\end{equation}
i.e., the sum of all positive elements of  $\mathbf D$ equals the sum of all negative elements of  $\mathbf D$. Moreover, $\sum_{i=0}^L  |D_i| /2$ is an integer between 0 and $K$.
\end{lemA}

\begin{proof}
For any feasible withdrawal vector $\mathbf y(n)$, we have from Equation (\ref{YVeq:cons2}) that
\[
\sum_{i=0}^L  y_i(n) = K,
\]
where $K$ is the number of servers. From equation (\ref{eq:I1}), we have then:
\begin{eqnarray*}
\sum_{i=0}^L  D_i &= & \sum_{i=0}^L  y^{MB}_i(n) - \sum_{i=0}^L  y_i(n) \\
&=& K-K = 0,
\end{eqnarray*}

\noindent and Equation (\ref{eq:I2}) follows. The last assertion of the lemma follows from

\begin{eqnarray*}
\sum_{i=0}^L  |D_i| &= & \sum_{i=0}^L  D_i \cdot \mathds 1_{\{D_i>0 \}} - \sum_{j=0}^L D_j \cdot \mathds 1_{\{D_j < 0 \}}  \\
&=& 2 \cdot  \sum_{i=0}^L  D_i \cdot \mathds 1_{\{D_i>0 \}}
\end{eqnarray*} 
and
\begin{eqnarray*}
\sum_{i=0}^L  |D_i| &= & \sum_{i=0}^L |y^{MB}_i(n) - y_i(n)|  \\
&\leq& \sum_{i=0}^L |y^{MB}_i(n)| +  \sum_{i=0}^L |y_i(n)| = 2K.
\end{eqnarray*}

\end{proof}

\begin{lemA} \label{lem:I41}
Consider a given state $(\mathbf x(n), \mathbf g(n))$ during time slot $n$. Let $f,t\in \{0,1,\ldots, L\}$ be any two queues such that $\mathbf I(f,t)$ is feasible. A policy $\pi \in \Pi$ that results in $ \hat x_t(n) \leq \hat x_f(n)-2$ does  not have the MB property at time $n$.
\end{lemA}

\begin{proof}
The interchange $\mathbf I(f,t)$  is a balancing interchange by definition. Since $ \hat x_t(n) \leq \hat x_f(n)-2$, then the balancing interchange $\mathbf I(f,t)$ reduces the imbalance index by a factor of two according to Equation (\ref{eq:I001}). Therefore, $\pi$ does not achieve the minimum imbalance index during time slot $n$, q.e.d.
\end{proof}


\begin{lemA} \label{lem:001}
Given the state $(\mathbf x(n),  \mathbf g(n) )$ and a feasible withdrawal vector $\mathbf y(n)$ then a  withdrawal vector $\mathbf y'(n)$ that results from performing any sequence of feasible, single-server reallocations on $\mathbf y(n)$ is feasible.
\end{lemA}

The proof of Lemma \ref {lem:001} is straightforward and therefore it is not included here.

\begin{lemA} \label{lem:002}
Consider the state $(\mathbf x(n),  \mathbf g(n) )$ and any two feasible withdrawal vectors $\mathbf y(n)$ and   $\mathbf y'(n)$. Then,  starting from $\mathbf y(n)$, the vector $\mathbf y'(n)$ can be obtained by performing a sequence of feasible, single-server reallocations.
\end{lemA}

\begin{proof}
To prove this lemma we  construct one such sequence next.

Let $\mathbf q(n),\mathbf q'(n) $ denote two server allocations for the implementation of $\mathbf y(n),\mathbf y'(n) $ respectively. Then we can relate $\mathbf y(n)$ and $\mathbf y'(n) $ as follows:

\begin{equation}\label{eq:0021}
 \mathbf y'(n) = \mathbf y(n) + \sum_{k=1}^K \mathbf I(q_k'(n),q_k(n))
\end{equation}

 \noindent since we assume both $\mathbf y(n)$ and $\mathbf y'(n) $ to be feasible,  server $k$ must be connected to both queues $q_k(n)$ and $q_k'(n)$. Therefore, each interchange $\mathbf I(q_k'(n),q_k(n))$ is equivalent to a feasible, single-server reallocation. Note that $q_k(n)=q_k'(n)$ is possible, for some $k$, in which case $\mathbf I(q_k'(n),q_k(n))=0$. By construction, all the interchanges in the right hand side of Equation (\ref{eq:0021}) are feasible.
\end{proof}

We are now ready to prove Lemma \ref{lem:I4} of Section \ref{section-lem:I4}.

\begin{proof}[Proof (Lemma \ref{lem:I4})]
We consider the following three cases assuming $f \neq t$:

Case 1: $f=0$.
This case is not possible by contradiction. By assumption, $D_0\geq +1$, which means that $y^{MB}_0(n) \geq y_0(n) +1$. This case states that an MB policy idled at least one more server than $\pi$. Therefore, $\hat x^{MB}_0(n) \leq -1$. This makes queue $0$ the shortest queue. Allocating the idled server to queue $t$, i.e., the  interchange $\mathbf I(t,0)$, is both feasible (since $\mathbf y(n)$ is feasible by assumption) and balancing (by Lemma \ref{lem:I003}). The interchange $\mathbf I(t,0)$ will result in a withdrawal vector $\mathbf y'(n) = \mathbf y^{MB}(n) + \mathbf I(t,0)$.
Let $s$ be the order of queue $f=0$ when ordering the vector $\mathbf{\hat x}^{MB}(n)$ in a descending manner. Therefore, $s=L+1$. Furthermore, in order for $\mathbf I(t,0)$ to be feasible queue $t$ must not be empty (according to feasibility constraint (\ref{eq:I102.1})) which implies that $\hat x^{MB}_t(n) \geq 1$ and the order of queue $t$ is $l<s$. Therefore, $\hat x^{MB}_f(n) \leq \hat x^{MB}_t(n)-2$ and the interchange $\mathbf I(t,0)$  will reduce the imbalance index by  $2 (s-l)$ according to Equation (\ref{eq:I001}). This implies that the new policy has a smaller imbalance index than an MB policy. This contradicts the fact that any MB policy minimizes the imbalance index.

Case 2: $t=0$.
%
When $t=0$ then the interchange $\mathbf I(f,t)$ is the process of allocating an idled server to queue $f>0$. This, according to Lemma \ref{lem:I003}, is a balancing interchange.

Case 3: $t,f>0$.
We will show that this case will also result in a balancing interchange.
Let $\mathbf y(n)$ be the original withdrawal vector. Let $\mathbf y^*(n)$ be the withdrawal vector resulted from the feasible interchange $\mathbf I(f,t)$, i.e.,
\[
\mathbf y^*(n)= \mathbf y(n) + \mathbf I(f,t)
\]

Using the assumption $D_t \leq -1$ and Equation (\ref{eq:30003}), we arrive at the following:
\begin{eqnarray}\label{eq:I41}
	y^{MB}_t(n) - y_t(n)   \hspace{-1mm} & \leq  \hspace{-1mm} & -1     \nonumber \\
	 y^{MB}_t(n)  \hspace{-1mm} & \leq  \hspace{-1mm} &   y^*_t(n) = y_t(n) -1
\end{eqnarray}
hence,
\begin{eqnarray}\label{eq:I42}
	x_t(n) - y^{MB}_t(n)   \hspace{-3mm} & \geq  \hspace{-3mm} &  x_t(n) -  y^*_t(n) = x_t(n) - ( y_t(n) -1) \nonumber \\
	 \hat x^{MB}_t(n)   \hspace{-3mm} & \geq  \hspace{-3mm} &  \hat x^*_t(n) =  \hat x_t(n) +1, \quad t>0
\end{eqnarray}

Similarly, using the assumption $D_f \geq +1$ and Equation (\ref{eq:30002}), we have
\begin{eqnarray}\label{eq:I43}
	y^{MB}_f(n) - y_f(n)  \hspace{-1mm} & \geq  \hspace{-1mm} & +1     \nonumber \\
	 y^{MB}_f(n) \hspace{-1mm} & \geq  \hspace{-1mm} &   y^*_f(n) = y_f(n) +1
\end{eqnarray}
hence,
\begin{eqnarray}\label{eq:I44}
	x_f(n) - y^{MB}_f(n)  \hspace{-3mm} & \leq  \hspace{-3mm} &  x_f(n) \!-  \! y^*_f(n) =  x_f(n) - ( y_f(n) +1) \nonumber \\
	 \hat x^{MB}_f(n)  \hspace{-3mm} & \leq  \hspace{-3mm} &  \hat x^*_f(n) =  \hat x_f(n) - 1, \quad f>0
\end{eqnarray}

To show that $ \mathbf I(f,t)$ in this case is a balancing interchange, we have to show that $\hat x_f(n) \ge \hat x_t(n) + 1$.  Suppose to the contrary that  $\hat x_f(n) \leq \hat x_t(n)$; then, from Equations (\ref{eq:I42}) and (\ref{eq:I44}), we have
\begin{eqnarray}\label{eq:I45}
	  \hat x_f(n) &\leq & \hat x_t(n)  \nonumber \\
	  \hat x^*_f(n) +1 & \leq & \hat x^*_t(n) -1  \nonumber \\
	  \hat x^*_f(n) & \leq & \hat x^*_t(n) -2
\end{eqnarray}

From (\ref{eq:I42}) and (\ref{eq:I44}) we have,
\begin{eqnarray}\label{eq:I46}
	  \hat x^{MB}_f(n) & \leq &  \hat x^*_f(n) \leq   \hat x^*_t(n) -2 \leq  \hat x^{MB}_t(n) -2 \nonumber \\
	  \hat x^{MB}_f(n) & \leq & \hat x^{MB}_t(n) -2
\end{eqnarray}

The differences satisfy the inequalities $D_f \geq +1$ and $D_t \leq -1$  by assumption, i.e., there is at least one more (respectively one less) server allocated to queue $f$ (respectively queue $t$) under the MB policy. Therefore, $y^*(n)= y^{MB}(n) + \mathbf I(t,f)$ is feasible and Inequality (\ref{eq:I46}) is a contradiction according to Lemma \ref{lem:I41}.
We must have that $\hat x_f(n) \ge \hat x_t(n) + 1$ and by definition the interchange $ \mathbf I(f,t)$  is a balancing interchange.
\end{proof}

\subsection{Proof of Lemma \ref{lem:I6}}\label{appendixC1i}

 
We present the proof of Lemma \ref{lem:I6} of Section \ref{section-lem:I4} in this appendix. Lemma \ref{lem:I11}, stated below, guarantees the existence of a feasible interchange.

\begin{lemA} \label{lem:I11}
Consider a given state $(\mathbf x(n), \mathbf g(n))$ and a  policy   $\pi \in \Pi_{n-1}$ that selects a withdrawal vector $\mathbf y(n)$  during time slot $n$.
Let $F, T $  denote the (non-empty) sets of queues for which $D_f\geq +1$ and $D_t\leq -1$, respectively.
Then,  there exist at least two queues $f\in F$ and $t\in T$ such that  the interchange $\mathbf I(f,t)$ is feasible.
\end{lemA}

\begin{proof}
Let $\pi^* \in \Pi^{MB}$ be an MB policy that selects  the  withdrawal vector $\mathbf y^*(n)$ during time slot $n$. Let $\mathbf D= \mathbf y^*(n)- \mathbf y(n)$. Furthermore, let $\mathbf q^*(n)$ and $\mathbf q(n)$ be two implementations of $\mathbf y^*(n)$ and $\mathbf y(n)$ respectively. From Lemma \ref{lem:002} we have:
\begin{equation}\label{eq:I1101}
 \mathbf y^*(n) = \mathbf y(n) + \sum_{k=1}^K \mathbf I(q_k^*(n),q_k(n))
\end{equation}

The summation on the right-hand side of Equation (\ref{eq:I1101}) is composed of $K$ terms, each of which represents a reallocation of a server $k$ from queue $q_k(n)$ to queue $q^*_k(n)$. Such server reallocation can be formulated as an interchange $\mathbf I(q^*_k(n), q_k(n)) $.

In the following, we will selectively use $i$ out of the $K$ terms of the summation  in Equation (\ref{eq:I1101}) to construct a feasible interchange $\mathbf I(r_1, r_{i+1})= \mathbf I(r_1,r_2) +\mathbf I(r_2,r_3) + \cdots + \mathbf I(r_{i},r_{i+1}) $ for some $i \leq K$, with $r_1 \in F$ and $r_{i+1} \in T$.
We will show that such a queue $r_{i+1}, i \leq K $ that belongs to $T$ does exist and the interchange   $\mathbf I(r_1, r_{i+1})$ is feasible.

Let $r_1 \in F, r_1 \in \{1,2, \ldots ,L \}$ then using Equation (\ref{eq:I1101}) we can write
\begin{equation}\label{eq:I1102}
 y^*_{r_1}(n) =  y_{r_1}(n) + \sum_{k=1}^K ( \mathds 1_{\{q_k^*(n)={r_1}\}} - \mathds 1_{\{q_k(n)={r_1}\}} )
\end{equation}
Since ${r_1} \in F$ by our assumption, then we have $D_{r_1} \geq 1$ and
\begin{equation}\label{eq:I1103}
 	y^*_{r_1}(n) -  y_{r_1}(n) \geq 1.
\end{equation}

From Equations (\ref{eq:I1102}) and  (\ref{eq:I1103}) we conclude
\begin{equation}\label{eq:I1104}
 	\sum_{k=1}^K \mathds 1_{\{q_k^*(n)={r_1}\}} \geq  \sum_{k=1}^K \mathds 1_{\{q_k(n)={r_1}\}} +1
\end{equation}

In  words, there is at least one more server allocated to queue ${r_1}$ under $\pi^*$ than the servers allocated to queue ${r_1}$ under $\pi$. Let $k_1$ be one such server. From  (\ref{eq:I1104}) we conclude that  one of the $K$ terms in Equation (\ref{eq:I1101}) must be $\mathbf I(q_{k_1}^*(n),q_{k_1}(n))$  such that  $q^*_{k_1}(n) ={r_1}, q_{k_1}(n) ={r_2}, k_1 \in \{1,2, \ldots K \}$. In other words, a server $k_1$ and  two queues ${r_1}=q_{k_1}^*(n)$ and ${r_2}=q_{k_1}(n)$ must exist such that the interchange $\mathbf I(r_1,r_2)$ is feasible.

The feasibility of $\mathbf I(r_1,r_2)$ stems from the fact that server $k_1$ is allocated to queues $r_1$ and $r_2$ under two different policies, namely $\pi^*$ and $\pi$. This is possible only if
\begin{equation}\label{eq:I1105}
 	g_{r_1,k_1}(n) =  g_{r_2,k_1}(n) = 1
\end{equation}
Furthermore, using Equation (\ref{eq:I1103}) we can write
\begin{eqnarray}\label{eq:I1106}
	y^*_{r_1}(n) &\geq &  y_{r_1}(n) + 1 \nonumber \\
	\text{or }\, x_{r_1}(n) - y^*_{r_1}(n) &\leq &  x_{r_1}(n) - y_{r_1}(n) - 1 \nonumber \\
	\text{hence, } \quad 0 \leq \hat x^*_{r_1}(n) &\leq &  \hat x_{r_1}(n) - 1
\end{eqnarray}

From Equation (\ref{eq:I1106}) we conclude
\begin{equation}\label{eq:I1107}
 	\hat x_{r_1}(n) \geq 1
\end{equation}

Equations (\ref{eq:I1105}) and  (\ref{eq:I1107}) are sufficient for the feasibility of the interchange $\mathbf y(n)+ \mathbf I(r_1,r_2)$.

Consider queue $r_2$ above. One of the following two cases may apply:

Case (1) $r_2 \in T$: The proof of the lemma in this case is completed by letting $f = r_1$ and $t= r_2$. The resulted interchange $\mathbf I(f,t), f \in F, t \in T$ is feasible by construction and the lemma follows.

Case (2) $r_2 \notin T$: Define $\mathbf y^1(n)$ as follows:
\begin{equation}\label{eq:I1108}
 \mathbf y^1(n) = \mathbf y(n) +  \mathbf I(r_1, r_2)
\end{equation}

Using Lemma \ref{lem:001} we conclude that $\mathbf y^1(n)$ is a feasible withdrawal vector. From Equations (\ref{eq:I1108}) and (\ref{eq:newI4}) we can write
\begin{equation}\label{eq:I1109}
  y^1_{r_2}(n) =  y_{r_2}(n) - 1
\end{equation}
From Equation (\ref{eq:I1109}) and since $r_2 \notin T$ we conclude
\begin{equation}\label{eq:I11010}
 y^*_{r_2}(n) \geq  y_{r_2}(n) \geq y^1_{r_2}(n) + 1
\end{equation}

Therefore, one of the terms of the summation in Equation (\ref{eq:I1101}) must be a server reallocation of some server $k_2$ from queue $r_3 = q_{k_2}(n)$ to queue  $r_2 = q^*_{k_2}(n)$, i.e., the interchange $\mathbf I(r_2, r_3)$ is a feasible, single-server reallocation. It follows that $\mathbf y^2(n) $ is feasible, where
\begin{eqnarray}\label{eq:I11011}
	 \mathbf y^2(n) &=& \mathbf y^1(n) +  \mathbf I(r_2, r_3) \nonumber \\
	 	&=& \mathbf y(n) +  \mathbf I(r_1, r_2)+  \mathbf I(r_2, r_3) \nonumber \\
	 	&=& \mathbf y(n) +  \mathbf I(r_1, r_3)
\end{eqnarray}

If $r_3 \in T$ then we complete the proof using the argument in case (1) above. Otherwise, we repeat the argument in case (2) again.

Repeating the previous argument $i$ times, $ 1 \leq i \leq K$,  we arrive at the following relationship:
\begin{equation}\label{eq:I11012}
 \mathbf y^i(n) = \mathbf y(n) +  \sum_{j=1}^i \mathbf I(r_j, r_{j+1}),
\end{equation}
where by construction, each one of the $i$ terms in Equation (\ref{eq:I11012}) above corresponds  uniquely to one of the terms of the summation in Equation (\ref{eq:I1101}). For every $i$ we check to see whether $r_{i+1} \in T$ (in which case the lemma is proved) or not. If not then we have
\begin{equation}\label{eq:I11015}
 y^*_{r_{i+1}}(n) \geq y_{r_{i+1}}(n) \geq  y^i_{r_{i+1}}(n) + 1
\end{equation}

Repeating the  argument $K$ times (one for each term of the summation in Equation (\ref{eq:I1101})) we will show that a queue $r_{i+1} \in T$, where $\mathbf I(r_i, r_{i+1})$ is one of the terms in Equation (\ref{eq:I1101}), must exist.

In order to do that, we assume to the contrary that  $r_{i+1} \notin T, \forall i=1,2, \ldots K$. The $K^{th}$ (last) server reallocation $\mathbf I(r_K, r_{K+1}), r_K= q^*_{k_K}(n), r_{K+1}= q_{k_K}(n)$ will result in the withdrawal vector $\mathbf y^K(n)$, such that,
\begin{equation}\label{eq:I11013}
 \mathbf y^K(n) = \mathbf y(n) +  \sum_{j=1}^K \mathbf I(r_j, r_{j+1})
\end{equation}
Since there is one-to-one correspondence between the summation terms in Equation (\ref{eq:I11013}) and those in Equation (\ref{eq:I1101}) by construction, then we can write
\begin{equation}\label{eq:I11014}
 \mathbf y^K(n) = \mathbf y(n) +  \sum_{k=1}^K \mathbf I(q^*_{k}(n), q_{k}(n)) 
\end{equation}
hence $\mathbf y^K(n) = \mathbf y^*(n)$. However, since $r_{K+1} \notin T$ then
\begin{equation}\label{eq:I11016}
 y^*_{r_{K+1}}(n)\geq y_{r_{K+1}}(n) \geq  y^K_{r_{K+1}}(n) + 1
\end{equation}
have a contradiction.
We conclude that there must exist a queue $r_{i+1} \in T$ such that server $k_i$ reallocation $\mathbf I(r_i, r_{i+1}), r_i= q^*_{k_i}(n), r_{i+1}= q_{k_i}(n)$ is feasible. Let $f=r_{1}$ and $t=r_{i+1}$. It follows that the interchange $\mathbf I(f,t)= \mathbf I(r_1, r_{i+1})$ is feasible and the lemma follows.
\end{proof}

\begin{proof}[Proof for Lemma \ref{lem:I6}]
From its definition and Lemma \ref{lem:I2}, $h_{\pi}  =\sum_{i=0}^L |D_i|/2$ is an integer between 0 and $K$. 
If $h_{\pi} = 0$,  then $\pi$ has the  MB property during time slot $n$ according to Lemma \ref{lem:I1}. 

So, suppose that $h_{\pi} > 0$. From Equation \ref{eq:I2}, we can pair queues $f_i$ and $t_i, 1\leq i \leq h_{\pi}$, such that for every $i$, $D_{f_i}\geq +1$ and $D_{t_i}\leq -1$. 
From Lemmas \ref{lem:I11} and \ref{lem:I4}, the interchange $\mathbf I(f_i,t_i)$ is feasible and balancing.   

Since we have $h_{\pi} $ such pairs of queues, then applying the $h_{\pi} $  balancing interchanges $\mathbf I(f_i,t_i)$ described by Lemma \ref{lem:I4} to policy $\pi$ will result in a policy $\pi^*$ for which $\mathbf D^*=0$, i.e., $\mathbf y^*(n)=\mathbf y^{MB}(n)$. Hence the resulting policy $\pi^* \in \Pi_n$.

\end{proof}



\subsection{Coupling Method and the Proof of Lemma \ref{lem:6}} \label{appendixB}


\subsubsection{The Coupling Method}

 
If we want to compare probability measures on a measurable space, it is often possible to construct random elements  \cite{Lindvall}, with these measures as their distributions, on a common probability space, such that the comparison can be carried out in terms of these random elements rather than the probability measures. The term \emph{stochastic coupling} (or coupling) is often used to refer to any such construction.
In the notation of  \cite{Lindvall}, a formal definition of coupling of two probability measures on the measurable space $(E, \mathcal{E})$ (the state space, e.g., $E=\mathcal{R},\mathcal{R}^d, \mathcal{Z}_+, etc.$) is given below; further details regarding coupling method and its application can be found in \cite{Lindvall}.

A random element in $(E, \mathcal{E})$ is a quadruple $(\Omega,\mathscr{F},\mathbf{P},X)$, where  $(\Omega,\mathscr{F},\mathbf{P})$ is the sample space and $X$ is the class of measurable mappings from $\Omega$ to $E$ ($X$ is an $E$-valued random variable, s.t. $X^{-1}(B) \in \mathscr{F}$ for all $B \in \mathcal{E} $).

\vspace{2mm}
\noindent\textbf{Definition:} \textit{A coupling of the two random elements $(\Omega,\mathscr{F},\mathbf{P},\mathbf{X})$ and $(\Omega',\mathscr{F}',\mathbf{P}',\mathbf{X}')$ in $(E, \mathcal{E})$ is a random element
$(\hat{\Omega},\hat{\mathscr{F}},\hat{\mathbf{P}},(\hat{\mathbf{X}},\hat{\mathbf{X}} ') )$ in ($E^2, \mathcal{E}^2$) such that}

\begin{equation} \label{def2:coupling}
     \mathbf{X} \eqd \hat{\mathbf{X}} \quad \text{ and } \quad \mathbf{X}'\eqd \hat{\mathbf{X}}',
\end{equation}
\textit{where $\eqd$ denotes 'equal in distribution'}.

\begin{rem}
 The above definition makes no assumption about the distribution of the collection of random variables  $\mathbf{X}$; for example, $\mathbf{X}$ may be a sequence of non-i.i.d. random variables.
\endremark
\end{rem}

In the area of optimal control of queues, coupling arguments have been used extensively to prove characteristics of the optimal policies for many queuing systems, c.f. \cite{Walrand}, \cite{Nain}, \cite{Tassiulas}, \cite{ganti} and many others.

We apply the coupling method to our proof as follows: Let $\omega$ and $\pi$ be a given sample path of the system state process and scheduling policy.  The values of the  sequences $\{X(n)\}$ and $\{Y(n)\}$ can be completely determined by $\omega$ and $\pi$. We denote the  ensemble of all random variables as system $S$. A new sample path, $\tilde \omega$ and a new policy $\tilde{\pi}$ are constructed as we specify in the proof.
We denote the  ensemble of all random variables (in the new construction) as system $\tilde S$. Then, in the  coupling definition, $\hat{\omega}=(\omega, \tilde{\omega})$ and the ``coupled'' processes of interest in Equation (\ref{def2:coupling}) will be the queue sizes $\hat{\mathbf{X}} = \{X(n)\}$ and $\hat{\mathbf{X}}' = \{\tilde X(n)\}$.

We define $\omega$  as the sequence of sample values of the random variables $(\mathbf{X}(1),\mathbf{G}(1), \mathbf{Z}(1),\mathbf{G}(2),\mathbf{Z}(2), \ldots)$, i.e., $ \omega \equiv ( {\mathbf{x}}(1), {\mathbf{g}}(1),  {\mathbf{z}}(1), {\mathbf{g}}(2), {\mathbf{z}}(2), \ldots)$. The sample path $\tilde \omega \equiv (\tilde {\mathbf{x}}(1),\tilde {\mathbf{g}}(1), \tilde {\mathbf{z}}(1),\tilde {\mathbf{g}}(2),\tilde {\mathbf{z}}(2), \ldots)$  is constructed such that (a) $\tilde {\mathbf{x}}(1)={\mathbf{x}}(1)$, (b)  $\tilde {\mathbf{g}}(n)$ is the same as ${\mathbf{g}}(n)$, except for  two elements that are exchanged,  (c) $ \tilde {\mathbf{z}}(n)$ is the same as ${\mathbf{z}}(n)$, except for  two elements that are exchanged. The selection of the appropriate elements that are exchanged is detailed in the proof that follows\footnote{In the  system we are studying, $\{\tilde {\mathbf{G}}(n), \tilde {\mathbf{Z}}(n) \}$ has the same distribution as $\{ {\mathbf{G}}(n),  {\mathbf{Z}}(n) \}$, since the distributions of ${\mathbf{G}}(n)$ as well as ${\mathbf{Z}}(n)$ are i.i.d. and will not change when their elements are reordered. The mappings from ${\mathbf{G}}(n)$ to $\tilde {\mathbf{G}}(n)$ and from ${\mathbf{Z}}(n)$ to $\tilde {\mathbf{Z}}(n)$ are one-to-one.}.

The new policy $\tilde{\pi}$ is constructed (by showing how $\tilde{\pi}$ chooses the withdrawal vector $\tilde{\mathbf{y}}(\cdot)$) as detailed in the proof. 
Then using Equation (\ref{sys_evol}), the new states $\mathbf{x}(\cdot), \mathbf{\tilde{x}}(\cdot)$ are determined under $\pi$ and $\tilde \pi$. The goal is to prove that the relation
\begin{equation} \label{prefered1}
    \mathbf{\tilde{x}}(t)\prec_p \mathbf{x}(t)
\end{equation}
is satisfied at all times $t$. Towards this end, the  preferred order (introduced in Section \ref{PreferredOrder-section}) can be described by the following property:

\noindent \textbf{Property 1:} $\mathbf{\tilde{x}}$ is preferred over $ \mathbf{x}  $  ($ \mathbf{\tilde{x}} \prec_p \mathbf{x} $) if and only if one of the  statements R1, R2 or R3 (that we introduced in Section \ref{PreferredOrder-section}) holds. We restate these statements here for the sake of presentation:
\begin{enumerate}
    \item[(R1)] $\mathbf{\tilde{x}} \leq \mathbf{x}$: the two vectors are component-wise ordered;
    \item[(R2)] $\mathbf{\tilde{x}}$ is a two-component permutation of $\mathbf{x}$ as described  in Section \ref{PreferredOrder-section}.
	 \item[(R3)] $\mathbf{\tilde{x}}$ is obtained from $\mathbf{x}$ by performing a \textsl{``balancing interchange"} as described  in Section \ref{PreferredOrder-section}.
\end{enumerate}



\subsubsection{Proof of Lemma \ref{lem:6} of Section \ref{YV:main-result}} ~~~
\label{YV:pi-tilde}

\begin{proof}
Fix an arbitrary policy $\pi \in \Pi_{\tau}^h$ and a sample path $\omega=(\mathbf{x}(1),\mathbf{g}(1),\mathbf{z}(1),\ldots )$, where $\mathbf{x}(.), \mathbf{g}(.)$ and $\mathbf{z}(.)$ are sample values of the random variables $\mathbf{X}(.),\mathbf{G}(.)$ and $\mathbf{Z}(.)$. Let , $\pi^* \in \Pi^{MB}$ be an MB policy that works on the same system. The policy $\pi^*$ chooses a withdrawal vector $y^*(t), \forall t$.

 The proof has two parts; Part 1 provides constructions for $\tilde \omega$ and $\tilde{\pi}$ (as defined by Lemma \ref{lem:6} statement) for times up to $t=\tau$. Part 2 does the same for $t>\tau$.

\vspace{2mm}\noindent\emph{\textbf{Part 1:}}
For the construction of $\tilde \omega$, we let the arrivals and channel states be the same in both systems at all time slots before $\tau$, i.e., $\mathbf{\tilde{z}}(t)=\mathbf{z}(t)$ and $\mathbf{\tilde{g}}(t)=\mathbf{g}(t)$ for all $t<\tau$. We construct $\tilde{\pi}$ such that  it chooses the same withdrawal vector
as $\pi$, i.e., we set $\tilde{ \mathbf{y}}(t) = \mathbf{y}(t)$ for all $t<\tau$. In this case, at $t = \tau$, the resulting queue sizes are equal, i.e., $\mathbf{\tilde{x}}(\tau)=\mathbf{x}(\tau)$.
In the remainder of Part 1, we will construct the policy $\tilde{\pi}$ at time slot $\tau$ such that: (a)  $\tilde{\pi} \in \Pi_{\tau}^{h-1}$, i.e., $\tilde{\pi}$ is closer to $\pi^* \in \Pi^{MB}$ than $\pi$, and (b) the resulting queue length under $\tilde{\pi}$ is preferred over that under the original policy ${\pi}$, i.e., $\mathbf{\tilde{x}}(\tau+1) \prec_p  \mathbf{x}(\tau+1)$.  Condition (b) is  necessary for proving the second part of the lemma, i.e., the domination of policy $\tilde{\pi}$ over ${\pi}$, a result that will be shown in Part 2 of this proof.

At time slot $\tau$,  let $\tilde{ \omega}$ have the same channel connectivities and arrivals as $\omega $, i.e., let $\mathbf{\tilde{g}}(\tau)=\mathbf{g}(\tau)$ and $\mathbf{\tilde{z}}(\tau)=\mathbf{z}(\tau)$. Furthermore, let $\mathbf D = y^*(\tau) - y(\tau)$. Recall that $h = \sum_{i=0}^L |D_i|/2$.
Then one of the following two cases may apply:

1- During time slot $t=\tau$, the original policy $\pi$ differs from $\pi^*$, the MB policy, by \emph{strictly} less than $h$ balancing interchanges. Then    $ \pi \in \Pi_{\tau}^{h-1}$ as well, so we  set $\tilde{ \mathbf{y}}(\tau) = \mathbf{y}(\tau)$. In this case, the resulting queue sizes $\mathbf{\tilde{x}}(\tau+1), \mathbf{x}(\tau+1)$ will be equal, property (R1)  holds true  and (\ref{prefered1}) is satisfied at $t=\tau+1$.

2- During time slot $t=\tau$, $\pi$ differs from the MB policy $\pi^*$ by \emph{exactly} $h$ balancing interchanges. Since $\pi \in \Pi_{\tau}^h$ and $h>0$ and following Lemma \ref{lem:I6}, we can identify  two queues $l$ and $s$ such that: (a) $D_l\geq 1$, (b) $D_s \leq -1$, and (c) $\mathbf I(l,s)$ is feasible.

The construction of $\tilde{\pi}$ is completed in this case by performing the   interchange $\mathbf I(l,s)$, i.e.,
\begin{equation}\label{eq:lem6interchange}
 	\mathbf {\tilde y}(\tau)= \mathbf { y}(\tau)+ \mathbf I(l,s),
\end{equation}
or equivalently,
\begin{equation}\label{eq:lem6interchange1}
 	\mathbf{ \tilde{\hat x}}(\tau)= \mathbf {\hat x}(\tau)- \mathbf I(l,s)
\end{equation}

According to Lemma \ref{lem:I4}, this interchange is balancing.
To complete the construction of $\tilde \omega$, we examine the arrivals under $\omega$ during time slot $\tau$. We set   $  {\tilde{z}_i}(\tau)= {z_i}(\tau), \forall i \neq l,s$. For queues $l$ and $s$,  we do the following: (i) if $x_l(\tau) = x_s(\tau) +1$ and  $  {z}_s(\tau) >  {z_l}(\tau)$ then we swap the arrivals for queues $l$ and $s$, i.e., we let $ {\tilde{z}_l}(\tau)= {z_s}(\tau)$ and $ {\tilde{z}_s}(\tau)= {z_l}(\tau)$, (ii) otherwise, let   $  {\tilde{z}_l}(\tau)= {z_l}(\tau)$ and  $  {\tilde{z}_s}(\tau)= {z_s}(\tau)$. 
The queue lengths at the beginning of  time slot ($\tau+1$) under the two policies  satisfy property (R1) in case (i) and (R3) otherwise.   
In either case, (\ref{prefered1}) is satisfied at $t=\tau+1$.

Starting from a preferred state at $t=\tau+1$, we will show next, in Part 2 of the proof, that a feasible control at time slot $t> \tau$,  such that the constructed policy $\tilde{\pi}$ dominates the original one $\pi$, will always exist. We do that by showing one such construction.

\vspace{2mm}\noindent\emph{\textbf{Part 2:}}
In this part of the proof, we construct  $\tilde{\omega}, \tilde{\pi}$ for times $t>\tau$, such that the  preferred order $\mathbf{\tilde{x}}(t) \prec_p \mathbf{x}(t)$ is valid for all $t>\tau$. This will insure  $\tilde{\pi}$ domination over $\pi$.
We will use induction to complete our proof. We assume that $\tilde{\pi}$ and $\tilde{\omega}$ are defined up to time $n-1$ and  that $\mathbf{\tilde{x}}(n) \prec_p \mathbf{x}(n)$. We will prove that at time slot $n$, $\tilde{\pi}$ can be constructed so that $\mathbf{\tilde{x}}(n+1) \prec_p \mathbf{x}(n+1)$. Thus, we have to show that either R1, R2 or R3 holds at time slot $n+1$.

The following three cases,  corresponding to properties (R1), (R2) and (R3) are considered next.

Case (1) $\mathbf{\tilde{x}}(n) \leq \mathbf{x}(n)$.  The construction of $\tilde \omega$ is straightforward in this case. We set  $\mathbf{\tilde{z}}(n)=\mathbf{z}(n)$ and $\mathbf{\tilde{g}}(n)=\mathbf{g}(n)$.  We construct $\tilde{\pi}$ such that   $\tilde{ \mathbf{y} }(n) = \mathbf{y}(n)$. In this case, its obvious that $\mathbf{\tilde{x}}(n+1) \leq \mathbf{x}(n+1)$ and  (\ref{prefered1}) holds at $t=n+1$.

Case (2) $\mathbf{\tilde{x}}(n)$ is a permutation of $\mathbf{x}(n)$, such that $\mathbf{\tilde{x}}(n)$ can be obtained from $\mathbf{x}(n)$ by permuting components $i$ and $j$ (as described in property R2 of the preferred order). For the construction of $\tilde \omega$,   we set $\tilde{g}_{i,c}(n) = g_{j,c}(n)$ and $\tilde{g}_{j,c}(n) = g_{i,c}(n)$, for all $c=1,2,\ldots,K$; $\tilde{z}_i(n) = z_j(n)$ and $\tilde{z}_j(n) = z_i(n)$ (refer to footnote 7 or \cite{ganti}); the connectivities and arrivals for each one of the remaining queues are the same as in $\omega$.
We construct $\tilde{\pi}$ such that $\tilde{y}_i(n) = y_j(n)$, $\tilde{y}_j(n) = y_i(n)$ and $\tilde{y}_m(n) = y_m(n)$ for all $m\neq i,j$. As a result, $\mathbf{\tilde{x}}(n+1)$ and $\mathbf{x}(n+1)$ satisfy  property  (R2)  and (\ref{prefered1}) is satisfied at $t=n+1$.

Case (3) $\mathbf{\tilde{x}}(n)$ is obtained from $\mathbf{x}(n)$ by performing a balancing interchange for queues $i$ and $j$ as defined in property (R3). In this case $x_i(n) \geq x_j(n)+1$, by the definition in (R3)\footnote{By definition, we have $x_i(n)> x_j(n)$, $\tilde x_i(n)= x_i(n)-1$ and $\tilde x_j(n)= x_j(n)+1$.}. There are three cases to consider:

(3.a) $x_i(n)= x_j(n)+1$. Therefore, $\tilde x_i(n)= x_j(n)$ and $\tilde x_j(n)= x_i(n)$,
i.e., the vectors $\mathbf{x}(n)$ and $\mathbf{\tilde{x}}(n)$ have components $i$ and $j$ permuted and all other components are the same. This case corresponds to case (2) above.

(3.b) $x_i(n)> x_j(n)+1$ and $y_i(n) \leq y_j(n)$.  We construct $\tilde \omega$ as in case (1) above, and we let   $\tilde{y}_m(n) = y_m(n), \forall m \neq j$. Note that it is not feasible for policy $\pi$ to empty queue $i$ in this case. Depending on whether $\pi$ empties queue $j$ or not at $t=n$, the construction of $\tilde{\pi}$  will follow one of the following two cases:


(i)  $y_j(n)<x_j(n)$, i.e., \underline{$\pi$ \textit{does not empty queue} $j$ at $t=n$}, then let $ \tilde y_j(n)=y_j(n)$ (i.e., $\tilde \pi$ is identical to $\pi$ at $t=n$). In this case, property (R3) will be preserved regardless of the arrivals pattern\footnote{Note that if $x_i = x_j+1$ then property (R2) is a special case of (R3).}, hence (\ref{prefered1}) is satisfied at $t=n+1$.

(ii)  $y_j(n)= x_j(n)$, i.e., \underline{$\pi$ \textit{empties queue} $j$ at $t=n$}.
Then if under policy $\pi$ all the servers connected to queue $j$  are allocated, then let $ \tilde y_j(n)=y_j(n)$. As in case (i) above, property (R3) holds and (\ref{prefered1}) satisfied at $t=n+1$.

In the event that $\pi$ empties queue $j$ without exhausting all the servers connected to queue $j$, then $\tilde \pi$ will be constructed such that one of these idling servers is allocated to queue $j$, i.e., $\tilde{y}_j(n) = y_j(n)+1$, so that $\tilde \pi$ preserves the work conservation property at $t=n$.
%
%
%
Since  $\tilde{x}_j(n) = x_j(n)+1$ by property (R3) and $\tilde{z}_j(n)=z_j(n)$ by construction, then we have
$$\tilde{x}_j(n+1) = x_j(n+1)= z_j(n)$$
Since $\tilde{x}_i(n) = x_i(n)-1$  by property (R3), $\tilde{z}_i(n)=z_i(n)$  and $\tilde{y}_i(n)=y_i(n)$ by construction, we have
$$\tilde{x}_i(n+1) = x_i(n+1)-1$$
The rest of the queues will have the same lengths in both systems at $t=n+1$. Therefore, (R1) holds with strict inequality and (\ref{prefered1}) is satisfied at $t=n+1$.  This case shows that a ``more'' balancing policy results in a strict enhancement of the original policy.

Cases (i) and (ii)  are the only possible ones, since $\pi$ cannot allocate more servers to queue $j$ than its length.



(3.c)  $x_i(n)> x_j(n)+1$ and $y_i(n) > y_j(n)$. We consider the following two cases:

(i)  $y_i(n)=x_i(n)$, i.e., \underline{$\pi$ \textit{ empties queue} $i$ at $t=n$}. 
To construct $\tilde \omega$ for this case,   we set $\mathbf{\tilde{z}}(n)=\mathbf{z}(n)$, $\tilde{g}_{m,c}(n)=g_{m,c}(n)$ for all $ m \neq i,j$, and for all $c$.
For queues $i$ and $j$ we do the following:

Let server $r$ be a server that is connected to queue $i$ at time slot $n$ such that $q_{r}(n)=i$ (i.e., server $r$ is allocated to queue $i$ by policy $\pi$ at $t=n$). Now,  we switch the connectivity of server $r$ to queue $i$ and that of server $r$ to queue $j$, i.e., we set $\tilde{g}_{j,r}(n)=g_{i,r}(n)$ and $\tilde{g}_{i,r}(n)=g_{j,r}(n)$ (refer to footnote 7 or \cite{ganti}). The rest of the servers will have the same connectivities to queues $i$ and $j$ under both policies, i.e.,  we set $\tilde{g}_{i,c}(n)=g_{i,c}(n)$ and $\tilde{g}_{j,c}(n)=g_{j,c}(n)$ for all $c \neq r$.


We construct $\tilde{\pi}$ such that $\tilde{q_r}(n)=j$ and $\tilde{q_c}(n)=q_c(n), \forall c\neq r$. This means that $\tilde{\pi}$ differs from $\pi$, at $t=n$, by one server allocation (server $r$) that is allocated to queue $j$ (under $\tilde{\pi}$) rather than queue $i$ (under $\pi$). From equation (\ref{Yi_qi}), we can easily calculate that the resulting queue lengths at $t=n+1$ (for any arrivals pattern) will be:
\begin{eqnarray*}
 \tilde{x}_{m}(n+1)   &=& x_{m}(n+1), \qquad \forall\, m. \nonumber
\end{eqnarray*}

It follows that property (R1) is satisfied and therefore (\ref{prefered1}) is  satisfied at $t=n+1$.

(ii)  $y_i(n)< x_i(n)$, i.e., \underline{$\pi$ \textit{does not empty queue} $i$ at $t=n$}. Then consider the following:

If $\pi$ does not empty queue $j$ at $t=n$ or if  $\pi$ empties queue $j$ and in the process it exhausts all   servers connected to queue $j$, i.e., $\pi$ does not idle any server connected to queue $j$, then we construct $\tilde \omega$ and $\tilde \pi$ similar to case (3.c(i)) above and the same conclusion holds. 
  
If on the other hand,  $\pi$ empties queue $j$ without exhausting all its connected servers and therefore $\pi$ is forced to idle some of the servers connected to queue $j$, then let  $r'$ be one such server. We set $\mathbf{\tilde{z}}(n)=\mathbf{z}(n)$, $\tilde{\mathbf g}(n)= \mathbf g(n)$. We construct $\tilde \pi$ such that $\tilde y_j(n) = y_j(n) +1$, by allocating server $r'$ to queue $j$ under $\tilde \pi$, i.e., we set $\tilde q_{r'}(n)=j$. This is feasible since $\tilde x_j(n) = x_j(n) +1$ by property (R3).  We also have  $\tilde x_i(n) = x_i(n) -1$  (by property (R3)).   Since $\mathbf{\tilde z}(n)= \mathbf z(n)$ by construction, then similar to case (3.b(ii)), property (R1) holds with strict inequality  at $t=n+1$ for any arrivals pattern, and   (\ref{prefered1}) follows.

Since $\pi$ cannot allocate more servers to queue $j$ than its length, therefore,  (i) and (ii)  are the only possible cases.

Note that policy $\tilde{\pi}$ belongs to $\Pi_{\tau}^{h-1}$ by construction in Part 1; its dominance over $\pi$  follows  from relation (\ref{func_class}).
\end{proof}

\section*{acknowledgement}{The  authors wish to thank professor L. Tassiulas 
from University of Thessaly, Volos, Greece  for his insightful comments throughout the course of this work.}

\end{document}